\documentclass[12pt,onecolumn]{IEEEtran}

\makeatletter

\let\proof\@undefined
\let\endproof\@undefined
\let\labelindent\@undefined
\makeatother

\usepackage{comment}

\usepackage{enumitem}
\usepackage{numprint}
\usepackage{graphicx}          
\usepackage{amssymb}
\usepackage{amsmath}
\usepackage{amsthm}
\usepackage[mathscr]{eucal}
\usepackage{cleveref}
\usepackage{pgfplots}
\usepackage{array}
\usepgfplotslibrary{fillbetween}
\pgfplotsset{compat=newest}
\usetikzlibrary{calc}

\usepackage{caption}
\usepackage{subcaption}

\crefname{equation}{}{}
\Crefname{equation}{Equation}{Equations}

\newtheorem{thm}{Theorem}
\crefname{thm}{Theorem}{Theorems}
\Crefname{thm}{Theorem}{Theorems}

\newtheorem{assum}{Assumption}
\crefname{assum}{Assumption}{Assumptions}
\Crefname{assum}{Assumption}{Assumptions}
\newtheorem{prop}{Proposition}

\crefname{figure}{Figure}{Figures}
\Crefname{figure}{Figure}{Figures}

\theoremstyle{definition}
\usepackage{thmtools}

\crefname{definition}{Definition}{Definitions}
\Crefname{definition}{Definition}{Definitions}

\theoremstyle{plain}
\declaretheorem[style=definition,qed=$\blacksquare$]{example}

\crefname{example}{Example}{Examples}
\Crefname{example}{Example}{Examples}

\newtheorem{corollary}{Corollary}
\theoremstyle{remark}
\newtheorem{rem}{Remark}
\theoremstyle{definition}

\newcommand{\ul}{\underline}
\newcommand{\ol}{\overline}
\newcommand{\td}{\widetilde}
\newcommand{\md}{\td}

\newcommand{\e}[1]{\mathscr{#1}}

\newcommand{\F}{\mathcal{F}}

\newcommand{\Rtsim}{\overset{\mathcal{R}_t}{\sim}}
\newcommand{\Rt}{\R_t}
\newcommand{\Rs}{\R_s}
\newcommand{\Rtau}{\R_{\tau}}
\newcommand{\RT}{\R_T}

\newcommand{\into}{\rightarrow}

\newcommand{\Rbb}{\mathbb{R}}
\newcommand{\Nbb}{\mathbb{N}}
\newcommand{\tr}{\top}
\newcommand{\I}{I}
\newcommand{\sm}{\backslash}
\newcommand{\set}[1]{\left\{#1\right\}}

\newcommand{\card}[1]{\#\left(#1\right)}
\newcommand{\paren}[1]{\left( #1 \right)}

\newcommand{\brak}[1]{\left[ #1 \right]}
\newcommand{\mat}[2]{\brak{\begin{array}{#1} #2 \end{array}}}
\newcommand{\norm}[1]{\left| #1 \right|}

\newcommand{\eqn}[1]{\begin{equation*}\begin{aligned} #1 \end{aligned}\end{equation*}}
\newcommand{\eqnn}[1]{\begin{equation}\begin{aligned} #1 \end{aligned}\end{equation}}

\newcommand{\spec}{\operatorname{spec}}
\newcommand{\diag}{\operatorname{diag}}

\newcommand{\cl}[1]{\ol{#1}}
\newcommand{\co}[1]{{#1}^C}
\newcommand{\nth}{ {\square\square}}
\newcommand{\rght}{ {\square\blacksquare}}
\newcommand{\lft}{ {\blacksquare\square}}
\newcommand{\bth}{ {\blacksquare\blacksquare}}
\newcommand{\uncon}{ {\rev{\square}}}
\newcommand{\con}  { {\rev{\blacksquare}}}

\newcommand{\LO}{\rev{\ensuremath{\blacksquare,\square}}}
\newcommand{\TD}{\rev{\ensuremath{\square,\blacksquare}}}

\definecolor{CornflowerBlue}{rgb}{0.258824,0.258824,0.435294}
\definecolor{cfblue}{rgb}{0.258824,0.258824,0.435294}
\definecolor{SkyBlue}{rgb}{0.196078,0.6,0.8}
\definecolor{TitleBlue}{RGB}{47, 67, 114}
\definecolor{dblue}{rgb}{.098,.243,.424}
\definecolor{lblue}{rgb}{.33,.57,.835}
\definecolor{llblue}{rgb}{.447,.643,.831}
\definecolor{lbluesam}{rgb}{.447,.643,.831} 
\definecolor{mblue}{rgb}{0.176, 0.380, 0.659}
\definecolor{lcomp}{rgb}{.969,.765,.416}
\definecolor{ddorange}{rgb}{0.624, 0.365, 0}
\definecolor{dorange}{rgb}{0.72, 0.506, 0.125}
\definecolor{lorange}{rgb}{0.961, 0.678, 0.165}
\definecolor{lgreen}{rgb}{.812,.969,.435}
\definecolor{dgreen}{RGB}{15,111,3}
\definecolor{mgreen}{RGB}{84, 174, 50}
\definecolor{lyellow}{rgb}{1,.859,.451}
\definecolor{dyellow}{rgb}{.651,.482,0}
\definecolor{lred}{rgb}{1,.6,.451}
\definecolor{dred}{rgb}{.65,.176,0}
\definecolor{dcompb}{RGB}{157,35,0}  
\definecolor{lcompb}{RGB}{186,70,30}
\definecolor{llcompb}{RGB}{255,136,92}
\definecolor{lcompbsam}{RGB}{255,136,92}  
\definecolor{dpurple}{RGB}{45,0,95}
\definecolor{mpurple}{RGB}{77,0,159}
\definecolor{lpurple}{RGB}{143,73,206}
\definecolor{purplea}{RGB}{122,24,207}
\definecolor{purpleb}{RGB}{68,17,112}
\definecolor{pptgreen}{RGB}{147, 208, 80}
\definecolor{darkred}{RGB}{200, 0, 0}

\newcommand{\rev}[1]{{#1}}

\renewcommand{\H}{\mathcal{H}}
\newcommand{\C}{\mathcal{C}}
\newcommand{\D}{\mathcal{D}}

\renewcommand{\F}{\mathcal{F}}
\newcommand{\G}{\mathcal{G}}
\newcommand{\R}{\mathcal{R}}

\newcommand{\mcM}{\mathcal{M}}
\newcommand{\mcMt}{\mathcal{M}_t}

\newcommand{\xjam}{x^\text{jam}}
\newcommand{\xcrit}{x^\text{crit}}
\newcommand{\ulx}{\underline{x}}
\newcommand{\J}{\mathcal{J}}

\usepackage[textwidth=1.5in]{todonotes}

\title{On infinitesimal contraction analysis \\ for hybrid systems} 
\author{Samuel A. Burden, Thomas Libby, and Samuel D. Coogan%
\thanks{S. A. Burden is and T. Libby was with the Department of Electrical \& Computer Engineering, University of Washington, Seattle, WA, USA, \texttt{sburden,tlibby@uw.edu}. S. D. Coogan is with the School of Electrical and Computer Engineering and the School of Civil and Environmental Engineering, Georgia Institute of Technology, Atlanta, GA, USA, \texttt{sam.coogan@gatech.edu}.  Support was provided by the U.~S. Army Research Laboratory and the U.~S. Army Research Office under contract/grant number W911NF-16-1-0158, by the Air Force Office of Scientific Research under award number FA9550-19-1-0015, and Awards \#1836819,1836932 from the National Science Foundation.
}%
\thanks{Notes on contributions:  SAB and SDC contributed the theoretical results and proofs; SDC contributed the application to transportation systems; SAB and TL contributed the application to mechanical systems.}
}

\date{}

\renewcommand\footnotemark{}

\begin{document}
\maketitle
\begin{abstract}
Infinitesimal contraction analysis, 
wherein global asymptotic convergence results are obtained from local dynamical properties,
has proven to be a powerful tool for applications in biological, mechanical, and transportation systems. 
\rev{
The technique has primarily been developed for systems governed by a single, possibly nonsmooth, differential or difference equation.
We generalize infinitesimal contraction analysis to hybrid systems governed by interacting differential and difference equations.
Importantly, we leverage an intrinsic distance function to derive the first contraction results for hybrid systems without restrictions on 
mode sequence
or 
dwell time.
}
Our theoretical results are illustrated in several examples and applications.

\end{abstract}

\section{Introduction}

A dynamical system is \emph{contractive} if all trajectories converge to one another~\cite{Lohmiller1998-xj}.
Contractive systems enjoy strong asymptotic properties, e.g. any equilibrium or periodic orbit is globally asymptotically stable.
Provocatively, these global results can sometimes be obtained by analyzing local (or \emph{infinitesimal}) properties of the system's dynamics.
In smooth differential (or difference) equations, for instance, a bound on a matrix measure (or induced norm) of the derivative of the equation can be used to prove global contractivity~\cite{Lohmiller1998-xj, Pavlov:2004lr, Sontag2010-qg, Aminzare:2014rm};
this approach has been successfully applied to biological~\cite{Raveh2016-ky, Russo2009-pg, Wang2005-zj}, mechanical~\cite{Lohmiller2000-kj,Manchester2014-kc}, and transportation~\cite{coogan2015compartmental, Como:2015ne} systems. 

Recent work has extended contraction analysis to certain classes of nonsmooth systems. Contraction for systems with a continuous vector field that is piecewise differentiable was first suggested in \cite{Lohmiller:2000kl} and rigorously characterized in \cite{di-Bernardo:2014oq}. 
Contraction of switched systems, 
potentially with sliding modes, is studied 
in \cite{Bernardo:2014nx} by explicitly considering contraction of the sliding vector field 
in \cite{Fiore:2016fj} via a regularization approach that does not require explicit computation of the sliding vector field. 
The paper \cite{Lu:2016yq} considers contraction of Carath\'{e}odory switched systems for which the time-varying switching signal is piecewise continuous and allows for different norms for each mode of the switched system.

The present paper complements and, in some cases, extends these prior works by considering a more general class
of \emph{hybrid} systems%
\footnote{We exclude Zeno behavior from consideration; %
this exclusion is informed by the applications and examples we seek to study.}
in which time-varying, state-dependent \emph{guards} trigger instantaneous transitions defined by \emph{reset} maps between distinct \emph{modes}. 
Different norms in each mode are allowed, and modes need not even be of the same dimension. 
\rev{
In contrast to previous work generalizing contraction analysis to hybrid systems, 
our generalization of infinitesimal contraction analysis 
does not restrict 
mode sequence as in~\cite{Tang2014-rd}
or 
dwell time as in~\cite{Tang2014-rd,Rifai2006-gt}.
}

At its core, the \emph{infinitesimal} approach to contractivity leverages local dynamical properties of continuous-time flow (or discrete-time reset) to bound the time rate of change of the distance between trajectories.
\rev{
This paper generalizes infinitesimal contraction analysis to hybrid systems by leveraging local dynamical properties of continuous-time flow \emph{and} discrete-time reset to bound the time rate of change of the \emph{intrinsic} distance between trajectories without imposing restrictions on 
mode sequence
or 
dwell time.
The intrinsic distance function we employ is derived 
in~\cite{Burden2015-ip} 
from the natural condition that the distance between a point in a guard and the point it resets to is zero. %
}%
This intrinsic distance function is distinct from the Skorohod~\cite{Gokhman2008-br} or Tavernini~\cite{Tavernini1987-wu} \emph{trajectory metrics}~\cite[Sec.~V-A]{Burden2015-ip} and from the \emph{distance function} introduced in~\cite{Biemond2016-ur}; it is an instantiation of the class of \emph{distance functions} defined in~\cite{Biemond2013-ph} we found particularly useful in the present context.
\rev{
Importantly, the use of this intrinsic distance function avoids restrictive conclusions regarding the closely-related notion of incremental stability~\cite[Prop.~1]{Biemond2018-dx}, a connection that will be discussed more fully in Section~\ref{sec:disc:inc} following the technical developments.
}

The conditions we obtain for infinitesimal contraction
(Theorem~\ref{thm:main} in Section~\ref{sec:main})
have intuitive appeal:
the derivative of the vector field, which captures the infinitesimal dynamics of continuous-time flow, must be infinitesimally contractive with respect to the matrix measure determined by the vector norm used in each mode~\eqref{eq:thm:mu};
the \emph{saltation  matrix}, which captures the infinitesimal dynamics of discrete-time reset%
\footnote{%
\rev{
The saltation matrix's role in contraction analysis of nonsmooth vector fields was first reported in~\cite{Fiore2017-au}.%
}
}, must be contractive with respect to the induced norm determined by the vector norms used on either side of the reset~\eqref{eq:thm:Xi}.
If upper and lower bounds on \emph{dwell time} are available, we can bound the intrinsic distance between trajectories, regardless of whether this distance is expanding or contracting in continuous- or discrete- time 
(Corollary~\ref{cor:main} in Section~\ref{sec:main}).
\rev{Furthermore, if the hybrid system is contractive with respect to the intrinsic distance function (defined in Section~\ref{sec:prelim:distance}), 
we show that the system is necessarily infinitesimally contractive in continuous and discrete time
(Theorem~\ref{thm:necessity} in Section~\ref{sec:necessity}).}
We present several examples 
(in Section~\ref{sec:main}) 
and applications 
(in Section~\ref{sec:appl}) 
to illustrate these theoretical contributions, and conclude with a discussion of our results and possible extensions (Section~\ref{sec:disc}).

\section{Notation}

Given a collection of sets $\set{S_\alpha}_{\alpha\in A}$ indexed by $A$, the \emph{disjoint union} of the collection is defined $\coprod_{\alpha\in A} S_\alpha=\bigcup_{\alpha\in A}(\{\alpha\}\times S_\alpha)$. 
Given $(a,x)\in\coprod_{\alpha\in A}S_\alpha$, we will simply write $x\in\coprod_{\alpha\in A}S_\alpha$ when $A$ is clear from context. 
For a function $\gamma$ with scalar argument, we denote limits from the left and right \rev{(when these exist)} by $\gamma(t^-)=\lim_{\sigma\uparrow t}\gamma(\sigma)$ and $\gamma(t^+)=\lim_{\sigma\downarrow t}\gamma(\sigma)$.
Given a smooth function $f : X \times Y \into Z$, 
we let 
$D_xf : TX \times Y \into TZ$
denote the 
\emph{derivative of $f$ with respect to $x\in X$}
and 
$Df = (D_x f, D_y f) : TX \times TY \into TZ$ 
denote the derivative of $f$ with respect to both $x\in X$ and $y\in Y$.
Here, $TX$ denotes the \emph{tangent bundle} of $X$; when $X\subset\Rbb^d$ we have $TX = X\times\Rbb^d$.
The \emph{induced norm} of a linear function $M:\Rbb^{n_{j}}\into\Rbb^{n_{j'}}$~is %
  $\|M\|_{j,j'}=\sup_{x\in \mathbb{R}^{n_{j}}}{|Mx|_{j'}}/{|x|_{j}}$
where $|\cdot|_j$ and $|\cdot|_{j'}$ denote the vector norms on $\mathbb{R}^{n_j}$ and $\mathbb{R}^{n_{j'}}$, respectively; 
when the norms are clear from context, we omit the subscripts. 
The \emph{matrix measure} of $A\in\Rbb^{n\times n}$, denoted $\mu(A)$, is %
\begin{align}
  \label{eq:51}
  \mu(A)=\lim_{h\downarrow 0}\frac{(\|I+hA\|-1)}{h}.
\end{align}

\section{Preliminaries}
\label{sec:prelim}
A \emph{hybrid system} is a tuple $\H=(\D,\F,\G,\R)$ where:
\begin{itemize}
\item[$\D$] $=\coprod_{j\in \J}\D_j$ 
  is a set of states
  where 
  $\J$ is a finite set of discrete states \rev{or \emph{modes}}
  and
  $\D_j=\mathbb{R}^{n_j}$ is a set of continuous states for each $j\in\J$ equipped with a norm $|\cdot|_j$ for some $n_j\in\mathbb{N}$;
\item[$\F$] $:[0,\infty)\times\D\to T\D$ is a time-varying vector field that we interpret as 
$\F_j=\left.\F\right|_{[0,\infty)\times\D_j}:[0,\infty)\times \D_j\to\mathbb{R}^{n_j}$ for each $j\in\J$;
\item[$\G$] $=\coprod_{j\in \J} \G_j$ is a time-varying guard set with $\G_j\subset [0,\infty)\times \D_j$ for all $j\in\J$;
\item[$\R$] $:\G\to\D$ is a time-varying reset map.
\end{itemize}

\noindent
We let $\G_{j,j'} = \G_{j}\cap\R^{-1}(\D_{j'})$ denote times and states in $\D_j$ that reset to $\D_{j'}$ for all $j,j'\in\J$, and given $t \geq 0$ we define $\G(t) = \G\cap\paren{\set{t}\times\D}$, $\G_j(t) = \G_j\cap\G(t)$, $\G_{j,j'}(t) = \G_{j,j'}\cap\G(t)$.
If a component such as $\F$, $\G$, or $\R$ is time-invariant, then we suppress time dependence in the corresponding notation.

\noindent
Before we assess infinitesimal contractivity, we first impose restrictions on the \emph{components} of the hybrid system as well as its \emph{flow}, that is, the collection of trajectories it accepts.
To help motivate and contextualize the assumptions,
we provide expository remarks 
following each assumption 
that explain how each condition is employed in what follows
and what specific dynamical phenomena it precludes,
and 
in Section~\ref{sec:disc:assum} we discuss why these assumptions are satisfied in several application domains of interest. 

\subsubsection{Hybrid system components and constructions}
\label{sec:prelim:components}
We begin by stating and discussing assumptions on the hybrid system components.

\begin{assum}[hybrid system components]
\label{assum:system}
For any hybrid system $\H = (\D,\F,\G,\R)$:
\begin{itemize}
\item[\ref{assum:system}.1] 
(vector field is Lipschitz and differentiable)
$\F_j = \F|_{\D_j}:[0,\infty)\times\D_j\into\D_j$ is globally Lipschitz continuous and continuously differentiable for all $j\in\J$;
\item[\ref{assum:system}.2] 
(discrete transitions are isolated)
$\R(t,\G(t))\cap\G(t) = \emptyset$ for all $t\in[0,\infty)$;
\item[\ref{assum:system}.3] 
(guards and resets are differentiable)
there exists continuously differentiable and nondegenerate%
\footnote{i.e. $D_x g_{j,j'}(t,x) \ne 0$ for all $t\in[0,\infty)$, $x\in\D_j$}
$g_{j,j'}:[0,\infty)\times\D_j\to \mathbb{R}$ such that $\G_{j,j'}(t)\subseteq \{x\in \D_j: g_{j,j'}(t,x)\leq 0\}\subseteq \G_{j}(t)$ 
and 
there exists continuously differentiable $\R_{j,j'}: \{(t,x) : x \in \D_j,\ g_{j,j'}(t,x)\leq 0\}\to \D_{j'}$ such that $\left.\R_{j,j'}\right|_{\G_{j,j'}(t)}=\left. \R \right|_{\G_{j,j'}(t)}$
for each $j,j'\in\J$ and $t \ge 0$ 
(whenever $\G_{j,j'}(t)\neq \emptyset$); 
\item[\ref{assum:system}.4] 
(vector field is transverse to guard)
$D_t g_{j,j'}(t,x) + D_x g_{j,j'}(t,x)\cdot \F_j(t,x) < 0$
for all $j,j'\in\J$, $t\geq 0$, and $x\in \G_{j,j'}(t)$.
\end{itemize}
\end{assum}

\noindent
Before we proceed, we make a number of remarks about the preceding Assumption.

\begin{rem}[vector fields generate differentiable global flows]
\label{rem:system:F}
Assumption~{\ref{assum:system}.1}
ensures there exists a continuously differentiable \emph{flow} $\phi_j : [0,\infty)\times[0,\infty)\times\D_j \into \D_j$ for $\F_j$.
In other words, if $x:[\tau,\infty)\into\D_j$ denotes the trajectory for $\F_j$ initialized at $x(\tau)\in\D_j$, then $x(t) = \phi(t,\tau,x)$ for all $t\in[\tau,\infty)$.
This condition enables application of classical infinitesimal contractivity analysis for continuous-time flows.
\end{rem}

\begin{rem}[discrete transitions are isolated without loss of generality]
\label{rem:system:R}
Since Assumption~\ref{assum:flow}.1 below will (in particular) prevent an infinite number of discrete transitions from occurring at the same time instant, 
Assumption~\ref{assum:system}.2 is imposed without loss of generality. 
Indeed, given a hybrid system that permitted at most $m$  discrete transitions at the same instant of time (an example with $m = 2$ can be found in~\cite[Thm.~8]{Johnson2016-nh}), the reset map could be replaced with its $m$-fold composition to yield a hybrid system with isolated discrete transitions that has the same set of trajectories (as defined below).
\end{rem}

\begin{rem}[guards are closed]
\label{rem:system:G}
Letting $\ol{\G}_{j,j'}(t) = \{x\in \D_j: g_{j,j'}(t,x)\leq 0\}$, and noting that $\ol{\G}_{j,j'}(t)$ is closed by continuity of $g_{j,j'}$, we observe that
$\G_j(t) = \bigcup_{j'\in\J} \G_{j,j'}(t) \subset \bigcup_{j'\in\J} \ol{\G}_{j,j'}(t) \subset \G_j(t)$,
whence 
Assumption~\ref{assum:system}.3 ensures
$\G_j(t) = \bigcup_{j'\in\J} \ol{\G}_{j,j'}(t)$ is a closed subset of $\D_j$ (once time dependence is dropped); guard closure is a crucial property used in the definition of the time-of-impact map in the next Remark.
(Note that the disjoint components of the guard, $\G_{j,j'}(t)$, are not required to be closed; this affordance will be helpful in the applications considered below.)
\end{rem}

\begin{rem}[time-of-impact is differentiable]
  \label{rem:system:tti}
  For all $j,j'\in\J$, let the \emph{time-of-impact} $\nu_{j,j'} : [0,\infty)\times \D_j\into[0,\infty]$ for the set $\ol{\G}_{j,j'}=\coprod_{t \geq 0} \ol{\G}_{j,j'}(t)$ containing $\G_{j,j'}$ be defined by 
\eqnn{\label{eq:tti}
\nu_{j,j'}(\tau,x) = \inf\set{t \ge \tau : g_{j,j'}(t,\phi_j(t,\tau,x)) \le 0},
}
with the convention that $\inf\emptyset = \infty$.
Assumption~\ref{assum:system}.3 
ensures $\nu_{j,j'}$ is well-defined, 
and
\rev{
Assumption~\ref{assum:system}.4 ensures $\nu_{j,j'}$ is continuously differentiable wherever it is finite;
these properties of the time-of-impact will play a crucial role in the proofs of Propositions~\ref{prop:trajectories} and~\ref{prop:salt} as well as Theorem~\ref{thm:main}.
}
\end{rem}

Before moving on, we use the time-of-impact maps associated with individual guards to define time-of-impact maps applicable to a mode's entire guard; these maps will subsequently be used to construct the hybrid system's trajectories.
For all $j\in\J$, let $\nu_j:[0,\infty)\times\D_j \into [0,\infty]$ be defined by
\eqnn{
\nu_j(\tau,x) = \min\set{\nu_{j,j'}(\tau,x) : j'\in \J}.
}

\subsubsection{Hybrid system trajectories and flow}
\label{sec:prelim:flow}
Informally, a {trajectory}%
\footnote{Since our analysis concerns infinitesimal contraction in continuous time, we deliberately avoid the concept of an \emph{execution}~\cite[Def.~II.3]{Lygeros2003-zx}, which is conventionally defined over a \emph{hybrid time domain}, that is, a set that indexes both continuous and discrete time.  
Formally, our \emph{trajectory} concept can be regarded as the right-continuous time parameterization (uniquely determined by Assumption~\ref{assum:system}.2) of the image of the corresponding execution.}
of a hybrid system is a 
right-continuous function of time that satisfies the continuous-time dynamics specified by $\F$ and the discrete-time dynamics specified by $\G$ and $\R$.
Formally, a function $\chi:[\tau,T)\to \D$ with $\tau\geq 0$ is a \emph{trajectory} of the hybrid system if:
\begin{enumerate}
\item $  D\chi(t)=\F(t,\chi(t))$ for almost all $t\in[\tau,T)$;
\item $\chi(t^+)=\chi(t)$ for all $t\in[\tau,T)$;
\item $\chi(t^-)= \chi(t)$ if and only if $\chi(t)\not\in \G(t)$;
\item whenever $\chi(t^-)\neq \chi(t)$, then $\chi(t^-)\in\G(t)$ and $\chi(t)=\R(t,\chi(t^-))$.
\end{enumerate}

\noindent
Note that it is allowed, but not required, that $T = \infty$ (although we will shortly impose additional assumptions that ensure trajectories are defined for all positive time).
If the domain of $\chi$ cannot be extended in forward time to define a trajectory on a larger time domain, then $\chi$ is termed \emph{maximal}.
The following Proposition ensures that 
maximal trajectories exist and are unique
under the conditions in  Assumption~\ref{assum:system}; its proof is standard~\cite[Thm.~III-1]{Lygeros2003-zx}.

\begin{prop}[existence and uniqueness of trajectories]
\label{prop:trajectories}
Under the conditions in Assumption~\ref{assum:system},
there exists a unique maximal trajectory $\chi:[\tau,T)\into\D$ satisfying $\chi(\tau)=x$ if $x\in\D\sm\G(\tau)$ or $\chi(\tau) = \R(\tau,x)$ if $x\in\G(\tau)$ for any initial state $x\in\D$ and initial time $\tau\geq 0$.
\end{prop}

\begin{proof}
Let $\tau \ge 0$ and $x\in\D_j$; if $x\in\G(\tau)$ then set $x = \R(\tau,x) \not\in\G(\tau)$.
If %
the time-to-impact map $\nu_j(\tau,x)=\infty$,
then the trajectory remains within $\D_j$ for all forward time, 
in which case we let $\chi|_{[\tau,\infty)}(s) = \phi_j(s,\tau,x)$ for all $s\in[\tau,\infty)$.
Otherwise, %
the trajectory flows to $\G(t)$ at time $t = \nu_j(\tau,x) > \tau$,
in which case we let $\chi|_{[\tau,t)}(s) = \phi_j(s,\tau,x)$ for all $s\in[\tau,t)$ and set $\chi(t) = \R(t,\phi_j(t,\tau,x))$.  
Applying the procedure described in the preceding sentences inductively 
from initial state $\chi(t)\in\D\sm\G(t)$ at initial time $t\ge 0$ uniquely determines $\chi$ at all times on a maximal interval $s\in[\tau,T)$, where $T < \infty$ if and only if the trajectory is \emph{Zeno}, that is, undergoes an infinite number of discrete transitions on the interval $[0,T)$.
\end{proof}

We will restrict the class of trajectories exhibited by the hybrid system in Assumption~\ref{assum:flow} below. 
Before imposing these restrictions, we first develop tools that enable analysis of how trajectories vary with respect to initial conditions.
At every time $t \ge 0$,
the restriction of the reset map 
$\Rt = \R|_{\G(t)}$ 
induces an equivalence relation $\Rtsim$ on $\D$ defined as the smallest equivalence relation containing 
\eqn{\set{ (x,y)\in \G(t) \times \D:\R(t,x) = y }\subset \D\times \D,} 
for which we write $x \Rtsim y$ to indicate $x$ and $y$ are related. The equivalence class for $x\in\D$ is defined as $[x]_{\Rt}=\set{y\in \D|x\Rtsim y}$. 
The time-varying \emph{quotient space} induced by the equivalence relation is denoted
\begin{align}
  \mcMt=\{[x]_{\Rt}|x\in \D\}
\end{align}
endowed with the quotient topology~\cite[Appendix~A]{Lee2012-mb};
we note that such quotient spaces have been studied repeatedly in the hybrid systems literature~\cite{Schatzman1998-gt, Simic2005-fv, Ames2006-uj, Burden2015-ip}.

To define a distance function on the quotient $\mcMt$, we will adopt the approach in~\cite{Burden2015-ip} and use the length of paths that are continuous in the quotient.
A path $\gamma:[0,1]\to \D$ is \emph{smoothly $\Rt$-connected} if
there exists an open set $\mathcal{O}\subset [0,1]$ such that:
the relative complement $\co{\mathcal{O}}\rev{\subset[0,1]}$ is countable (so that, in particular, the closure $\cl{\mathcal{O}}=[0,1]$);
$\gamma$ is continuously differentiable on $\mathcal{O}$; 
and 
$\lim_{r'\uparrow r} \gamma(r')\Rtsim \lim_{r'\downarrow r} \gamma(r')$ for all $r\in(0,1)$ and $\gamma(0) = \lim_{r'\downarrow 0}\gamma(r')$, 
$\gamma(1) = \lim_{r'\uparrow 1}\gamma(r')$. 

A set $\mathcal{O}$ satisfying the above conditions is termed a \emph{support set} for $\gamma$ at time $t$.
Intuitively, a smoothly $\Rt$-connected path $\gamma$ is a path through the modes $\{\D_j\}_{j\in \J}$ of the hybrid system that is allowed to jump through the reset map $\Rt$ 
(forward or backward) and is smooth almost everywhere. With a slight abuse of notation,%
\footnote{Formally, $\pi_t\circ\gamma$ is a path in $\mcMt$, where $\pi_t:\D\to\mcMt$ is the \emph{quotient projection}.}
we consider $\gamma$ a path in $\mcMt$. 
With this identification, all $\Rt$-connected paths are (more precisely:  descend to) continuous paths in the quotient space $\mcMt$.
Any support set $\mathcal{O}$ for a smoothly $\Rt$-connected path $\gamma$ 
is a countable union of (disjoint) open intervals; %
let $\mathcal{O} = \bigcup_{i=1}^k (u_i,v_i)$ with possibly $k=\infty$.
Because each segment $\left.\gamma\right|_{(u_i,v_i)}$ is continuously differentiable, 
the segment is (in particular) continuous, 
so its image must necessarily belong to a single $\D_j$ for some $j\in\J$.

\begin{assum}[hybrid system flow]
\label{assum:flow}
For any hybrid system $\H = (\D,\F,\G,\R)$:
\begin{itemize}
\item[\ref{assum:flow}.1] 
(Zeno) 
no trajectory undergoes an infinite number of resets in finite time;
\item[\ref{assum:flow}.2] 
(continuity of hybrid system flow) 
with $\phi:\e{F}\into\D$ denoting the hybrid system flow, i.e. $\phi(t,\tau,x) = \chi(t)$ where%
\footnote{Note that Assumption~\ref{assum:flow}.1 ensures that the maximal time interval in Proposition~\ref{prop:trajectories} is $[\tau,\infty)$, i.e. $T = \infty$.}
$\chi:[\tau,\infty)\into\D$ is the unique trajectory initialized at $\chi(\tau) = x$, 
the projection 
$\pi_t\circ\phi(t,\tau,x)$, 
regarded as a function $[0,\infty)\times\D\into\mcMt$,
is continuous;
\item[\ref{assum:flow}.3] 
(support sets of smoothly $\Rt$-connected paths) 
for all $t\geq \tau \geq 0$ and all smoothly $\Rtau$-connected paths $\gamma\in\Gamma(\tau)$, 
$\phi(t,\tau,\gamma(\cdot))$ is smoothly $\Rt$-connected and has a support set
$\mathcal{O}$ such that, for all $r'\in\mathcal{O}$, 
there exists $\epsilon>0$ 
such that all trajectories $\phi(\cdot,\tau,\gamma(r))$ with $r\in (r'-\epsilon,r'+\epsilon)$ undergo the same sequence of discrete state transitions as $\phi(\cdot,\tau,\gamma(r'))$ on the time interval $[\tau,t]$.
\end{itemize}
\end{assum}

\noindent
Before we proceed, we make a number of remarks about the preceding Assumption.

\begin{rem}[Zeno \rev{and forward completeness}]
\label{rem:zeno}
Since our results below will strongly leverage the fact that the hybrid system flow is everywhere locally a composition of a finite number of differentiable flows and resets, we cannot easily extend our approach to accommodate Zeno trajectories.
\rev{
By excluding Zeno, we ensure that all trajectories exist uniquely on infinite forward time horizons, i.e.\ $T = \infty$ in Proposition~\ref{prop:trajectories}.
}
\end{rem}

\begin{rem}[continuity of hybrid system flow]
Trajectories cannot be infinitesimally contractive wherever the flow is discontinuous.
To see this, note that the construction of the path-length distance function $d_t$ below and its compatibility with the topology of the hybrid quotient space $\mcM_t$ ensures the distance between trajectories on either side of a discontinuity will grow linearly with time.
\end{rem}

\begin{figure}
\centering{

\begin{minipage}[b]{1.0\columnwidth}
  \centering
{
  \def\svgwidth{6cm} 
\resizebox{!}{2.6cm}{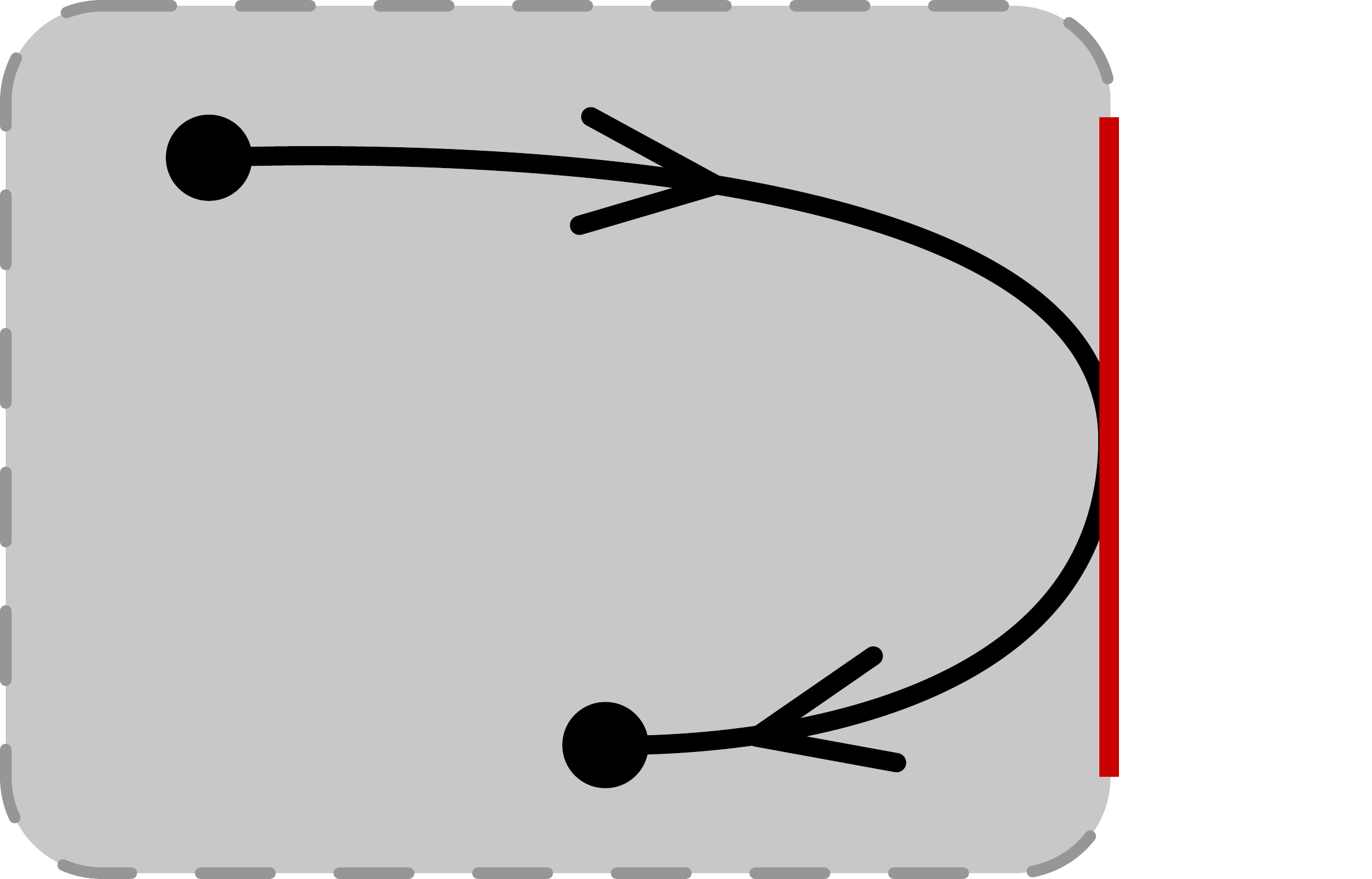}
  }
  \subcaption{{\bf disallowed:}  grazing}\label{fig:assum:grazing}
\end{minipage}%
\\
\begin{minipage}[b]{\columnwidth}
  \centering
{
  \def\svgwidth{13cm} 
  \resizebox{!}{2.6cm}{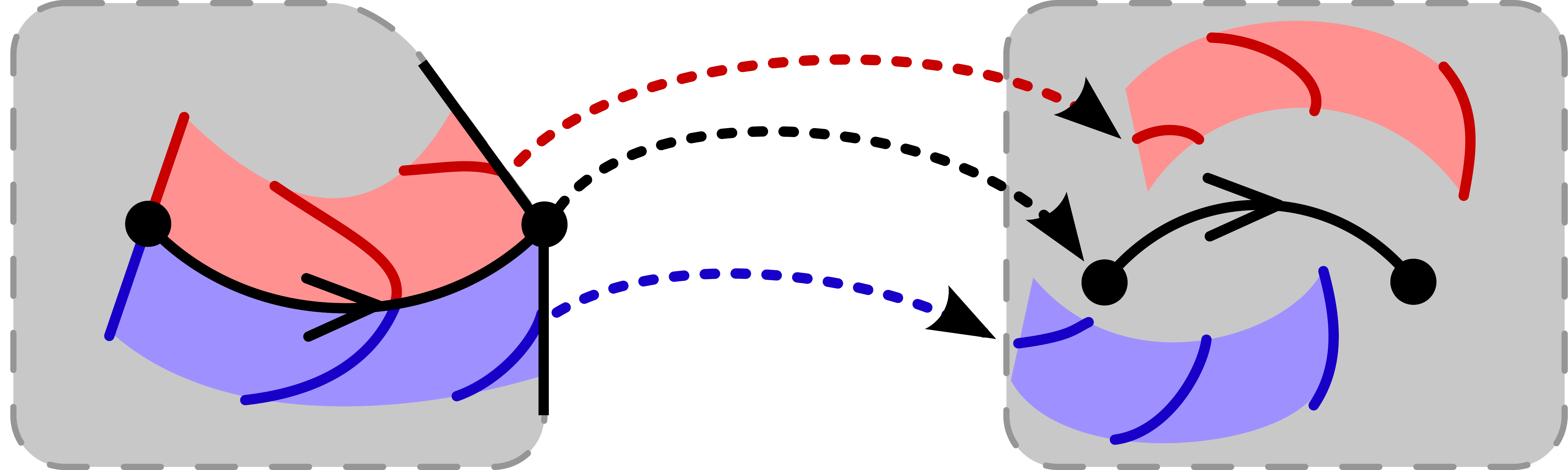}
  }
  \subcaption{{\bf disallowed:}  discontinuous trajectory outcomes}\label{fig:assum:discont}
\end{minipage}%
\\
\begin{minipage}[b]{1.0\columnwidth}
  \centering
{
  \def\svgwidth{10cm} 
  \resizebox{!}{2.6cm}{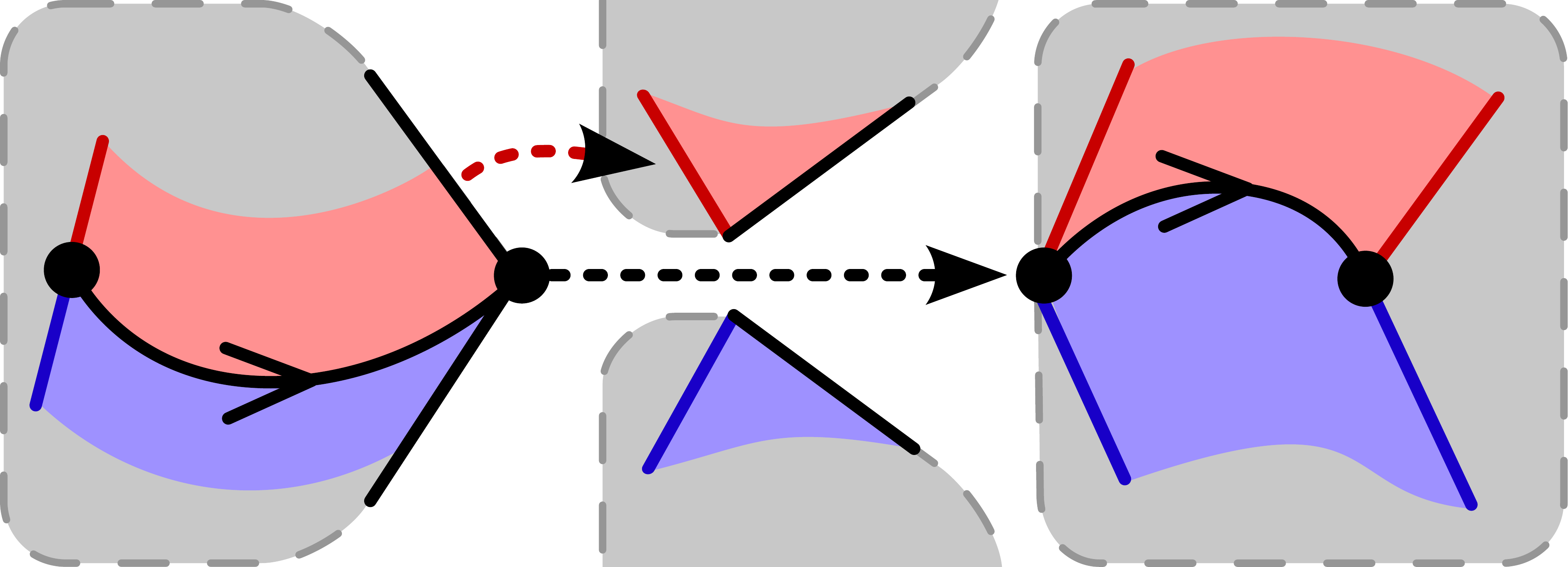}
  }
  \subcaption{{\bf allowed:}  overlapping guards with continuous trajectory outcomes}\label{fig:assum:cont}
\end{minipage}%
}
\caption{%
Illustration of some dynamical phenomena that are (dis)allowed by Assumptions~\ref{assum:system} and~\ref{assum:flow}.
{\bf Disallowed:}  
trajectory intersects guard without undergoing reset, i.e. \emph{grazing} (1.4, \emph{vector field is transverse to guard});
trajectory outcomes depend discontinuously on initial conditions (2.2, \emph{continuity of hybrid system flow}).
{\bf Allowed:}  
nonempty intersection of guard closures, i.e. overlapping guards;
continuous and piecewise-differentiable trajectory outcomes;
\rev{
trajectories that visit different sequences of modes;
}
different norms used in different modes,
which can have different dimensions.%
}
\label{fig:assum}
\end{figure}

\begin{rem}[support sets of $\Rt$-connected paths] 
\label{rem:support}
Note that 
Assumption~\ref{assum:system} already suffices to ensure the conditions in 
Assumptions~\ref{assum:flow}.2-\ref{assum:flow}.3  hold
when and where guards do not ``overlap'',
i.e. where the intersection of their closures is empty, $\cl{\G}_{j,j'}\cap\cl{\G}_{j,j''} = \emptyset$. %
When and where guards do overlap ($\cl{\G}_{j,j'}\cap\cl{\G}_{j,j''} \ne \emptyset$), %
Assumption~\ref{assum:flow}.3
does not require that \emph{all} trajectories along a path undergo the same sequence of discrete state transitions, only that the path's domain contains an open dense subset wherein each connected component undergoes the same sequence of discrete state transitions on finite time horizons. %
\rev{
We emphasize that this condition \emph{does not} require that all trajectories visit the same sequence of modes as illustrated in~Fig.~\ref{fig:assum}(c).
}
\end{rem}

It is well-known~\cite{Aizerman1958-ih} 
that, 
under favorable conditions, %
the hybrid system flow $\phi$ is differentiable almost everywhere and, 
moreover, 
its derivative can be computed by solving a jump-linear-time-varying differential equation.
The preceding assumptions %
are favorable enough to ensure the flow has these properties so that, 
in particular, 
the derivative along a path can be computed using 
the jump-linear-time-varying differential equation.
These facts are summarized in the following Proposition, whose proof is standard~\cite{Aizerman1958-ih}.

\begin{prop}
\label{prop:salt}
Under Assumptions~\ref{assum:system} and~\ref{assum:flow},
given an initial time $\tau\geq 0$ 
and 
a smoothly $\R_\tau$-connected path
$\gamma\in\Gamma_\tau$, 
let $\psi(t,r)=\phi(t,\tau,\gamma(r))$ for all $t\ge \tau$  
and define
\begin{align}
  \label{eq:24}
  w(t,r)=D_r\psi(t,r)
\end{align}
whenever the derivative exists. 
Then 
$w(\tau^-,r)=D_r\gamma(r)$ 
and 
$w(\cdot,r)$ satisfies a linear-time-varying differential equation
\begin{align}
\label{eq:variational}
  D_t{w}(t,r)=D_x \F(t,\psi(t,r))w(t,r),\ \psi(t^-,r)\in \D\backslash\G(t),
\end{align}
with jumps 
\begin{align}
\label{eq:saltation}
  w(t,r)=\Xi(t,\psi(t^-,r)) w(t^-,r),\ \psi(t^-,r)\in \G(t),
\end{align}
where $\Xi(t,x)$ is a \emph{saltation matrix} given by \eqref{eq:saltmat} at the top of the next page for all $t\geq 0$ and all $x\in \G_{j,j'}(t) = \G{_j}(t)\cap\R^{-1}(\D_{j'})$.
\newcounter{MYtempeqncnt}
\begin{figure*}[!t]
\normalsize
\eqnn{
\label{eq:saltmat} \Xi(t,x)=& D_x\R(t,x) +
\frac{\left(\F_{j'}(t,\R(t,x))-D_x\R(t,x)\cdot \F_{j}(t,x) - D_t\R(t,x) \right)\cdot D_xg_{j,j'} (t,x)}{D_t g_{j,j'}(t,x) + D_x g_{j,j'}(t,x)\cdot \F_{j}(t,x)}
}
\hrulefill
\vspace*{4pt}
\end{figure*}
\end{prop}

\begin{proof}[Proof]
Given $r\in[0,1]$,
suppose that the trajectory $\chi$ 
initialized at $\chi(\tau) = \gamma(r)$
undergoes exactly one discrete transition 
on the interval $[\tau,t]$
through $\G_{j,j'}(\sigma)$ at time 
$\sigma\in(\tau,t)$ 
so that $\psi(\sigma^-,r) = \chi(\sigma^-)\in\G_{j,j'}(\sigma)$.
Assumption~\ref{assum:flow}.3 implies that, generically, all trajectories initialized sufficiently close to $\chi(\tau) = \gamma(r)$ also undergo exactly one discrete transition through $\G_{j,j'}$ on the interval $[\tau,t]$.
Then for all $s\in[\tau,\sigma)$ we may write
$\psi(s,r) = \phi_j(s,\tau,\gamma(s))$,
whence 
$D_s \psi(s,r) = \F_j(s,\psi(s,r))$
and hence
\begin{align}
D_s D_r \psi(s,r) 
& = 
D_r D_s \psi(s,r) \\
& = D_r \F_j(s,\psi(s,r)) \\
& = D_x \F_j(s,\psi(s,r)) D_r \psi(s,r),
\end{align}
so~\eqref{eq:variational} is satisfied for times in $[\tau,\sigma)$.
Similarly, for all $s\in[\sigma,t]$ we may write
\begin{align}
\nonumber \psi(s,r) =  \phi_{j'}(&s,\nu_{j,j'}(\tau,\gamma(r)), \R(\nu_{j,j'}(\tau,\gamma(r)),\\
\label{eq:psi}&\phi_j(\nu_{j,j'}(\tau,\gamma(r)),\tau,\gamma(r)))
\end{align}
where $\nu_{j,j'}(\tau,x)$ denotes the time-of-impact for guard $\G_{j,j'}$ for the trajectory initialized at state $x$ at time $\tau$,
whence 
$D_s \psi(s,r) = \F_{j'}(s,\psi(s,r))$
and hence
\begin{align}
D_s D_r \psi(s,r) 
& = 
D_r D_s \psi(s,r) \\
& = D_r \F_{j'}(s,\psi(s,r)) \\
& = D_x \F_{j'}(s,\psi(s,r)) D_r \psi(s,r),
\end{align}
so~\eqref{eq:variational} is satisfied for times in $[\sigma,t]$.
Each function that appeared above is continuously differentiable wherever its derivative is evaluated since:
Assumption~\ref{assum:system}.1 implies $\phi_{j'}$ and $\phi_{j}$ are continuously differentiable;
Assumption~\ref{assum:system}.3 implies $\R|_{\G_{j,j'}}$ is continuously differentiable;
and
Assumption~\ref{assum:system}.4 implies the time-to-impact %
is continuously differentiable.

It remains to be shown that $D_r \psi(\sigma,r)$ is related linearly to $\lim_{s\uparrow\sigma} D_r \psi(s,r)$ via a \emph{saltation matrix} of the form in~\eqref{eq:saltmat}.
To see that this is the case, note that
\begin{align}
\nonumber \psi(\sigma+\varepsilon,r) = & \phi_{j'}
(
\sigma+\varepsilon,
\nonumber \nu_{j,j'}(\sigma-\varepsilon,\psi(\sigma-\varepsilon,r)),  \\
\nonumber& \hspace{18pt} \R(\nu_{j,j'}(\sigma-\varepsilon,\psi(\sigma-\varepsilon,r)),\\
\nonumber&\hspace{30pt}\phi_j(\nu_{j,j'}(\sigma-\varepsilon,\psi(\sigma-\varepsilon,r)),\\
&\hspace{45pt}\sigma-\varepsilon,\psi(\sigma-\varepsilon,r))
).
\end{align}
Differentiating both sides of this expression with respect to $r$ and taking the limit as $\varepsilon\downarrow 0$,
\begin{align}
\nonumber &D_r \psi(\sigma,r) \\
\nonumber & = 
\left[
D_2 \phi_{j'} \cdot D_2 \nu_{j,j'} \right. \\
\nonumber & \quad  + D_3 \phi_{j'} \cdot( D_1\R \cdot D_2\nu_{j,j'} \\
&\quad + D_2\R \cdot [ D_1 \phi_j \cdot D_2 \nu_{j,j'} + D_3\phi_j])]
\cdot D_r \psi (\sigma^-,r) \\
\nonumber & = \left[ \vphantom{\frac{1}{2}} D_2\R(\sigma,x) \right. \\
\nonumber & \quad \left. + \big( \vphantom{\frac{1}{2}} \F_{j'}(\sigma,\R(\sigma,x)) - D_2\R(\sigma,x)\cdot \F_{j}(\sigma,x) - D_1\R(\sigma,x)\big) \right. \\
\nonumber & \quad\quad \left. \times \left(\frac{D_2g_{j,j'}(\sigma,x)}{D_1g_{j,j'}(\sigma,x) + D_2g_{j,j'}(\sigma,x) \cdot \F_{j}(\sigma,x)} \right) \right]\\ &\quad \cdot D_r \psi (\sigma^-,r),
\end{align}
where we have made use of the following substitutions that apply when we evaluate both time arguments at $\sigma$ and the state argument at $x = \psi(\sigma^-,r)$:
$D_2\phi_{j'}(\sigma,\sigma,\R(\sigma,x)) = -\F_{j'}(\sigma,\R(\sigma,x))$,
$D_1\phi_{j}(\sigma,\sigma,x) = \F_{j}(\sigma,x)$,
$D_3\phi_{j'}(\sigma,\sigma,\R(\sigma,x)) = I$,
$D_3\phi_{j}(\sigma,\sigma,x) = I$,
\begin{align}
D_2\nu_{j,j'}(\sigma,x) = \frac{-D_2g_{j,j'}(\sigma,x)}{D_1g_{j,j'}(\sigma,x) + D_2g_{j,j'}(\sigma,x) \cdot \F_{j}(\sigma,x)}.
\end{align}
Thus $D_r \psi(\sigma,r)$ is related linearly to $D_r \psi(\sigma^-,r)$ via a \emph{saltation matrix} of the form in~\eqref{eq:saltation}, as desired.

Assumption~\ref{assum:flow}.1 ensures there are a finite number of discrete transitions on the (bounded) interval $[\tau,t]$,
whence the preceding argument can clearly be applied inductively to accommodate all discrete transitions on the interval $[\tau,t]$. 
\end{proof}

\subsubsection{Hybrid system intrinsic distance}
\label{sec:prelim:distance}
As the final preliminary construction,
we define the length of a smoothly $\Rt$-connected path $\gamma:[0,1]\into\D$ as the sum of the lengths of its segments, and use this \emph{length structure}~\cite[Ch.~2]{Burago2001-sj} to derive a distance function on $\mcMt$.
To that end, 
define the length of a (continuously differentiable) path segment $\left.\gamma\right|_{(u_i,v_i)}:(u_i,v_i)\into\D_j$ in the usual way using the norm $|\cdot|_j$ in $\D_j$, 
namely,
$L_j(\left.\gamma\right|_{(u_i,v_i)})=\int_{u_i}^{v_i}|D\gamma_j(r)|_jdr$ (we drop the subscript for $L$ when the mode is clear from context), %
and define length of $\gamma$ as the sum of the lengths of its segments,
\begin{align}
\nonumber L(\gamma)&=\sum_{i=1}^kL(\left.\gamma\right|_{(u_i,v_i)})=\int_{\mathcal{O}}|D\gamma(r)|_{j(r)} dr \\
  \label{eq:19} &=\int_0^1|D\gamma(r)|_{j(r)} dr,
\end{align}
where $j(r)\in \J$ denotes the mode satisfying $\gamma(r)\in \D_{j(r)}$ for each $r\in[0,1]$.
(Note that the value of the expression $L(\gamma)$  does not depend on the particular support set $\mathcal{O}$, 
so the length of $\gamma$ is well-defined.)
With $\Gamma(t)$ denoting the set of  smoothly $\Rt$-connected paths in $\mcMt$, 
and letting
\begin{align}
  \label{eq:31}
  \Gamma(t,x,y)=\{\gamma\in\Gamma(t):\gamma(0)=x\text{ and }\gamma(1)=y,\ x,y\in\D\}
\end{align}
denote the subset of paths that start at $x\in\D$ and end at $y\in\D$,
we define the distance $d_t(x,y)$ between $x$ and $y$ at time $t$ 
by
\begin{align}
  \label{eq:d}
  d_t(x,y)=\inf_{\gamma\in\Gamma(t,x,y)}L(\gamma).
\end{align}
\rev{%
We note that the function $d_t:\mcMt\times \mcMt\to \mathbb{R}_{\geq 0}$
belongs to the class of \emph{distance function} specified in~\cite[Def.~1]{Biemond2013-ph}, but it has a stronger \emph{intrinsic} relationship to the hybrid system:
$d_t$ is a distance function on $\mcMt$ compatible with the quotient topology~\cite[Thm.~13]{Burden2015-ip}.
}

\begin{rem}[distance function]
Note that the distance function defined in~\eqref{eq:d} is time-varying if (and only if) the guard $\G$ and/or the reset $\R$ are time-varying.
Although this property may initially seem counter-intuitive, we argue that it is in fact natural that the intrinsic distance varies with time in hybrid systems that have time-varying guards and/or resets.%
\footnote{Of course, time-varying metrics have long been employed in contraction analysis~\cite{Lohmiller1998-xj}.}
Indeed, consider a toy example with 
$\D = \D_1\coprod\D_2$
where
$\D_1 = \D_2 = \Rbb$, 
i.e. the system state consists of two distinct copies of the real line $\Rbb$, 
and
suppose the state only transitions from $\D_1$ to  $\D_2$ through the guard $\G_{1,2}(t) = \set{(1,x)\in\D_1 : x = t}$,
at which point the reset function leaves the continuous state unaffected, 
i.e. $(1,x) \mapsto (2,x)$.
Then, according to~\eqref{eq:d}, the distance between state $(1,x)\in\D_1$ and $(2,x)\in\D_2$ equals 
$d_t\paren{(1,x),(2,x)} = 2(x-t)$ 
for $x > t$ and zero for $x \le t$.
This calculation captures the intuition that the distance between $(1,x)$ and $(2,x)$ decreases as the guard translates up the number line.
\end{rem}

\section{Main result}
\label{sec:main}
The main contribution of this paper is the provision of local (or \emph{infinitesimal}) conditions under which the distance between any pair of trajectories in a hybrid system
(as measured by the intrinsic function defined in \eqref{eq:d})
is globally bounded by an exponential envelope. 
These conditions are made precise in Theorem \ref{thm:main} and Corollary \ref{cor:main}. 
In what follows, we will only check infinitesimal contractivity conditions on a 
\emph{contraction region} $\C\subset\D$,
that is,
a forward-invariant subset ($\xi\in\C$ implies $\phi(t,0,\xi)\in\C$ for all $t\ge 0$)
that descends to a simply-connected subset in the quotient. 

\begin{thm}
\label{thm:main}
Under Assumptions~\ref{assum:system} and~\ref{assum:flow},
if 
there exists $c\in\mathbb{R}$ %
\rev{and contraction region $\C\subset\D$ %
} 
such that
  \begin{align}
    \label{eq:thm:mu}
\mu_j\left(D_x \F_j(t,x)\right)\leq c
  \end{align}
  for all $j\in \J$, $x\in\rev{\C\cap}\D_j\backslash \G_j$, $t\geq 0$,
and
  \begin{align}
    \label{eq:thm:Xi}
\|\Xi(t,x)\|_{j,j'}\leq 1
  \end{align}
  for all $j,\rev{j'}\in\J$, $x\in\rev{\C\cap}\G_{j,j'}$, $t\geq 0$, 
then
  \begin{align}
    \label{eq:thm:d}
 d_t( \phi(t,0,\xi),\phi(t,0,\zeta))\leq e^{ct} d_0(\xi,\zeta)
  \end{align}
for all $t\geq 0$ and $\xi,\zeta\in \C\cap\D$.
\end{thm}
\begin{proof}
Given $x(0)=\xi$, $z(0)=\zeta$, \rev{and $s\in[0,t]$}, for fixed $\epsilon>0$, let $\gamma:[0,1]\to D$ be a piecewise-differentiable $\Rs$-connected path satisfying $\gamma(0)=\xi$, $\gamma(1)=\zeta$, and  $L(\gamma)<d_0(\xi,\zeta)+\epsilon$, and let $\psi(t,r)=\phi(t,0,\gamma(r))$. Since $\phi(t,0,\cdot)$ is piecewise-differentiable, it follows from Assumption \ref{assum:flow}.2 and \ref{assum:flow}.3 that $\psi(t,\cdot)$ is a piecewise-differentiable $\Rt$-connected path for all $t\geq 0$. Let $w(t,r)=D_r\psi(t,r)$ whenever the derivative exists. By Proposition~\ref{prop:salt}, $w(t,r)$ satisfies the jump-linear-time-varying equations
\begin{align}
\label{eq:800}
  \dot{w}(t,r)=\frac{\partial f}{\partial x}(t,\psi(t,r))w(t,r),\ \psi(t,r)\in\D\sm\G,
\end{align}
\begin{align}
\label{eq:7}
  w(t^+,r)=\Xi(t,\psi(t^-,r))w(t^-,r),\ \psi(t^-,r)\in \G.
\end{align}

We claim that 
\begin{align}
\label{eq:thm:w}
  |w(t,r)|\leq e^{ct}|w(0^-,r)|
\end{align}
for all $t\geq 0$ and for all $r\in[0,1]$ whenever $w(t,r)$ exists. To prove the claim, for fixed $r$, let $\{t_i\}_{i=1}^k\subset [0,\infty)$ with $t_0\leq t_1\leq \cdots$ and  possibly $k=\infty$ be the set of times at which the trajectory $\phi(t,s,\gamma(r))$ intersects a guard so that $\left.\psi(\cdot,r)\right|_{[t_{i},t_{i+1})}$ is continuous for all $i\in\{0,1,\ldots,k-1\}$ where $t_0=s$ by convention, and, additionally, $\left.\psi(\cdot,r)\right|_{[t_{k},\infty)}$ is continuous if $k<\infty$. Note that if $k=\infty$ then $\lim_{i\to\infty} t_i=\infty$ since Zeno trajectories are not allowed.
Now consider some fixed time $T>0$. If $k<\infty$ and $t_k\leq T$, let $i=k$; otherwise, let $i$ be such that $t_i\leq T <t_{i+1}$. 
Let $j$ be the active mode of the system during the interval $[t_i,t_{i+1})$, i.e. $\psi(t,r)\in \D_j$ for all $t\in[t_i,t_{i+1})$.  
With $J(t)=D_x\F_j(\psi(t,r))$ for $t\in[t_i,t_{i+1})$  
we have
\begin{align}
\label{eq:4}
  |w(T,r)|&\leq e^{\int_{t_{i}}^T\mu(J(\tau))d\tau}|w(t_{i}^+,r)|\\
\label{eq:4-2} &\leq e^{c(T-t_{i})}|w(t_{i}^+,r)|\\
\label{eq:4-3}&\leq e^{c(T-t_{i})}\|\Xi(t,\psi(t_i^-,r))\||w(t_{i}^-,r)|\\
\label{eq:4-4}&\leq e^{c(T-t_{i})}|w(t_{i}^-,r)|
\end{align}
where 
\eqref{eq:4} follows from Coppel's inequality applied to \eqref{eq:800},  
\eqref{eq:4-2} follows from \eqref{eq:thm:mu}, 
\eqref{eq:4-3} follows from \eqref{eq:7}, 
and 
\eqref{eq:4-4} follows from \rev{\eqref{eq:thm:Xi}}; 
see, e.g., \cite[p. 34]{Desoer:2008bh} for a characterization of Coppel's inequality. 
Since \eqref{eq:4}--\eqref{eq:4-4} holds for any $T<t_{i+1}$, we further conclude that $|w(t_{i+1}^-,r) |\leq e^{c(t_{i+1}-t_{i})}|w(t_{i}^-,r)|$ whenever $i\leq k$. Then, by recursion, $  |w(T,r)|\leq e^{cT}|w(s^-,r)|$. Since $T$ was arbitrary,~\eqref{eq:thm:w} holds. 

Again fix $T> 0$. Because $\psi(T,\cdot)$ is a piecewise-differentiable $\RT$-connected path, there exists a support set $\mathcal{O}=\bigcup_{i=1}^k (u_i,v_i)$ of $\psi(T,\cdot)$  
such that $\left.\psi(T,\cdot)\right|_{(u_i,v_{i})}$ is continuously-differentiable for all $i\in\{0,1,\ldots, k\}$. 
It follows that
\begin{align}
  \label{eq:10}
  L\left(\left.\psi(T,\cdot)\right|_{(u_i,v_i)}\right)= \int_{u_i}^{v_i} |w(T,\sigma)| d\sigma.
\end{align}

Then 
\begin{align}
\label{eq:12}
  L(\psi(T,\cdot))&=\sum_{i=1}^{k}   L\left(\left.\psi(T,\cdot)\right|_{(u_i,v_i)}\right)\\
&\leq \int_0^1|w(T,\sigma)|d\sigma\\
\label{eq:12-2}&\leq e^{cT}\int_0^1|w(0^-,\sigma)|d\sigma\\
\label{eq:12-3}&=e^{cT}L(\gamma)\\
&\leq e^{cT}(d_0(\xi,\zeta) + \epsilon)
\end{align}
where \eqref{eq:12-2} follows from \eqref{eq:thm:w}, and \eqref{eq:12-3} follows because $w(0^-,r)=D_r\gamma(r)$. In addition, observe
\begin{align}
  \label{eq:14}
 d_T( \phi(T,0,\xi),\phi(T,0,\zeta))\leq L(\psi(T,\cdot)).
\end{align}
Since $T$ was arbitrary and $\epsilon$ can be chosen arbitrarily small, \eqref{eq:thm:d} holds.
\end{proof}

\begin{example}
\label{ex:1}
Consider a \rev{planar} hybrid system with two modes in the positive orthant of the plane so that $\D=\D_{L}\coprod\D_{R}$ with $\D_{L}=\D_{R}=\{x\in\mathbb{R}^2|x_1\geq 0\text{ and }x_2\geq 0\}$, and further take $g _{R,L} (x)=x_1-1$ and $g _{L,R} (x)=1-x_1$ so that the system is in the left (resp., right) mode $\D_{L}$ (resp., $\D_R$) when $x_1< 1$ (resp., $x_1>1$). Assume the reset map $\R$ is the identity map and $\dot{x}=\F_j(x)=A_jx$ for $j\in\{L,R\}$ with
\begin{align}
  \label{eq:54}
  A_j=
  \begin{bmatrix}
    -a_j&0\\
0&-b_j
  \end{bmatrix},\quad a_j,b_j>0\quad \text{for }j\in\{L,R\}.
\end{align}
All trajectories initialized in $\D$ flow to $\D_L$ and converge to the origin.  
Equip both modes with the standard Euclidean 2-norm so that $|x|_L=|x|_R=|x|_2$ and consider two trajectories $x(t)=\phi(t,0,\xi)$, $z(t)=\phi(t,0,\zeta)$ with initial conditions $\xi,\zeta\in\D$. 
Then $d(x(t),z(t))=|x(t)-z(t)|_2=|e(t)|_2$ for $e(t)=x(t)-z(t)$. When both trajectories are in the same mode so that $x,z\in\D_j$ for some $j\in\{L,R\}$, the error dynamics obey the dynamics of that mode. It therefore follows that $D_td(x,z)\leq \max\{-a_j,-b_j\} d(x,z)$
so that the distance decreases at exponential rate $\max\{-a_j,-b_j\} $.

Now suppose $x$ and $z$ are in different modes at some time $t$ and, without loss of generality, assume $x\in \D_L$ and $z\in \D_R$. 
Writing
\begin{align}
  \label{eq:1}
  x=\begin{bmatrix}
1-\epsilon_L\\
x_2
\end{bmatrix},\quad z=
  \begin{bmatrix}
    1+\epsilon_R\\
x_2+\delta
  \end{bmatrix}
\end{align}
for some $\epsilon_L,\epsilon_R>0$ and $\delta\in\mathbb{R}$, we have
\begin{align}
  \label{eq:3}
D_t(d(x,z)^2)&= D_t((x-z)^T(x-z))\\
\nonumber&=2(\epsilon_L+\epsilon_R)(a_L-a_R)+2\delta x_2\rev{(b_L-b_R)}\\
&\quad +\text{H.O.T.}
\end{align}
where the higher order terms H.O.T. are quadratic in $\epsilon_L$, $\epsilon_R$, and $\delta$. Then $D_t(d(x,z)^2)<0$  for all $x_2\geq 0$ and all sufficiently small $\epsilon_L>0$, $\epsilon_R>0$, $\delta\in\mathbb{R}$ if and only if $a_L<a_R$ and $b_L=b_R$. In other words, contraction between any two arbitrarily close trajectories transitioning from $\D_R$ to $\D_L$ occurs only if trajectories ``slow down'' in the direction normal to the guard surface when transitioning modes, and the dynamics orthogonal to the guard are unaffected. This example is illustrated in Fig.~\ref{fig:ex1}.

\end{example}

\begin{figure}
  \centering
{\footnotesize
    \begin{tikzpicture}
  \begin{axis}[width=2.2in,
    height=1.5in,
   axis y line=left,
   axis x line=bottom,
every axis x label/.style={at={(rel axis cs:1,0)},anchor=west},
every axis y label/.style={at={(rel axis cs:0,1)},anchor=south},
   xmin=.5, xmax=1.5, ymax=1,ymin=0, 
   ytick=\empty,
   xtick={.5,1},
   xticklabels={$1-\epsilon$,1},
    xlabel={$x_1$},
    ylabel={$x_2$},
    clip=true,
    at={(0,0)},
    ]
    \draw[dotted] (1,0) -- (1,.8);
    \node at (.8,.1) {$\D_L$};
    \node at (1.2,.1) {$\D_R$};
    \node[circle, fill=black, draw=black,inner sep=1pt] (a) at (.7,.6) {};
    \node[circle, fill=black, draw=black,inner sep=1pt] (b) at (1.3,.5) {};
    \node[circle, fill=black, draw=black,inner sep=1pt] (a2) at ($(.63, {.41*(exp(1.3*.63)-1)})$) {};
    \node[circle, fill=black, draw=black,inner sep=1pt] (b2) at ($(1.1, {.73*(exp(.4*1.1)-1)})$) {};

    \node[inner sep=1pt,label=90:$x(t)$] (a3) at ($(.77, {.41*(exp(1.3*.77)-1)})$) {};
    \node[,inner sep=1pt,label=90:$z(t)$] (b3) at ($(1.4, {.73*(exp(.4*1.4)-1)})$) {};
    \node[inner sep=1pt] (a4) at ($(.76, {.41*(exp(1.3*.76)-1)})$) {};
    \node[,inner sep=1pt] (b4) at ($(1.39, {.73*(exp(.4*1.39)-1)})$) {};
    \addplot[black,domain=0:.82,dashed] (\x,{.41*(exp(1.3*\x)-1)});
    \addplot[black,domain=0:1.5,dashed] (\x,{.73*(exp(.4*\x)-1)});
    \draw[->,>=latex,line width=1pt] (b) --node[above,pos=.3]{$e(t)$} (a);
    \draw[->,>=latex,line width=1pt] (b2) --node[below=1pt,pos=.8]{$e(t+\tau)$} (a2);
    \draw[->,line width=1pt] (b3)--(b4);
    \draw[->,line width=1pt] (a3)--(a4);
\end{axis}
\end{tikzpicture}
}
\caption{An illustration of two trajectories $x(t)$ and $z(t)$ of the hybrid system in Example \ref{ex:1} in different modes $\D_L$ and $\D_R$. The distance between trajectories is the Euclidean length of $e(t)=x(t)-z(t)$. When $x(t)$ and $z(t)$ are close, $|e(t)|$ decreases over a short time window $[t,t+\tau]$ if and only if $a_L<a_R$ and $b_L = b_R$, that is, the horizontal component of $x$ decreases at a slower rate than the horizontal component of $z$ and the rates of change of the vertical components are equal.}
  \label{fig:ex1}
\end{figure}
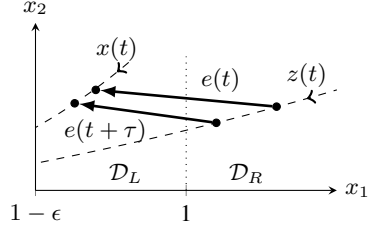

Now suppose that uniform upper or lower bounds on the \emph{dwell time} between successive resets are known. 
Then the number of discrete state transitions is upper or lower bounded on any compact time horizon, and the proof of Theorem \ref{thm:main} can be adapted to derive an exponential bound on the intrinsic distance between any pair of trajectories.
\rev{
Specifically, the bound obtained in~\eqref{eq:thm:w} must be augmented by the quantity $\max\{K^{\lceil t/\underline{\tau}\rceil},K^{\lfloor (t-s)/\overline{\tau}\rfloor}\}$ that accounts for the effect of discrete resets,
}%
as in the following Corollary.

\begin{corollary}
\label{cor:main}
Under Assumptions~\ref{assum:system} and~\ref{assum:flow},
suppose the dwell time between resets is at most $\overline{\tau}\in (0,\infty]$ and at least $\underline{\tau}\in[0,\infty)$
\rev{and there exists forward-invariant $\C\subset\D$ %
}
such that
  \begin{equation}\label{eq:cor:mu}
  \mu_j(D_x\F_j(t,x))\leq c
  \end{equation}
for some $c\in\mathbb{R}$ for all $j\in\J$, $x\in \rev{\C\cap}\D_j\backslash \G_j$, $t\geq 0$, and
\begin{equation}\label{eq:cor:Xi}
  	\|\Xi(t,x)\|_{j,j'}\leq K	
\end{equation}
for some $K\in\mathbb{R}_{\geq 0}$ and all $j\in\J$, $x\in\rev{\C\cap}\G_{j,j'}$, $t\geq 0$.
Then, for all $t\geq s\geq 0$ \rev{and $\xi,\zeta\in\C\cap\D$,}
\begin{align}
\nonumber & d_t( \phi(t,0,\xi),\phi(t,0,\zeta))\\
&\leq \max\{K^{\lceil t/\underline{\tau}\rceil},K^{\lfloor (t-s)/\overline{\tau}\rfloor}\}e^{c(t-s)} d_s(\xi,\zeta).
\end{align}
In particular, if $\max\{Ke^{c\underline{\tau}},Ke^{c\overline{\tau}}\}<1$ then 
\begin{equation}
\label{eq:cor:zero}
\lim_{t\to\infty}d_t(\phi(t,s,\xi),\phi(t,s,\zeta))=0.
\end{equation}
\end{corollary}

We now address an important special case, namely, when modes have the same dimension, are equipped with the same norm, and resets are simply translations (e.g., identity resets). 
The first part of the following Proposition establishes that the induced norm of the saltation matrix is lower bounded by unity; 
in the particular case of the standard Euclidean 2-norm, 
the second part of the  Proposition shows that 
the induced norm of the saltation matrix is equal to unity if and only if the difference between the vector field evaluated at $x$ and $\R(x)$ lies in the direction of the gradient of the guard function.

\begin{prop}
\label{prop:main}
  Under Assumptions~\ref{assum:system} and \ref{assum:flow}, 
  suppose:
  the guard set $\G$ and the reset map $\R$ are time-invariant;
  $\R_{j,j'}$ is a translation for some $j,j'\in\J$ 
  (i.e., $D\R_{j,j'}(x)=I$ for all $x\in \G_{j,j'}$); 
  and 
  $|\cdot|_j=|\cdot|_{j'}$. 
  Then the following properties hold:
  \begin{enumerate}
  \item $\|\Xi(t,x)\|_{j,j'}\geq 1$ for all $x\in\G_{j,j'}$, and all $t\geq 0$;
\item If $|\cdot|_j=|\cdot|_{j'}=|\cdot|_2$ where $|\cdot|_2$ denotes the standard Euclidean 2-norm, then $\|\Xi(t,x)\|_{j,j'}= 1$ if and only if  
\begin{align}
   \label{eq:26}
\F_{j'}(t,\R(x))-\F_{j}(t,x) = \alpha(t,x) Dg_{j,j'} (x)^T
 \end{align}
for all $t\geq 0$ and all $x\in\G_{j,j'}$ for some $\alpha:[0,\infty)\times \G_{j,j'}\to \mathbb{R}$ satisfying 
\begin{align}
  \label{eq:450}
0\leq \alpha(t,x)\leq \frac{-2Dg_{j,j'}(x)\cdot\F(t,x)}{|Dg_{j,j'} (x)|_2^2}.
\end{align}
  \end{enumerate}
\end{prop}
\begin{proof} 
\ 

  \begin{enumerate}
  \item Under the hypotheses of the proposition, we have that
  \begin{align}
\Xi(t,x)
= 
I+\frac{\left(\F_{j'}(t,\R(x))-\F_{j}(t,x)\right)\cdot Dg_{j,j'} (x)}{Dg_{j,j'}(x)\cdot \F_{j}(t,x)}.
  \end{align}
Fix $x\in \G_{j,j'}$ and let $z\in\text{Null}(Dg_{j,j'}(x))$. Then $\Xi(t,x) z=z$ so that always $\|\Xi(t,x)\|_{j,j'}\geq 1$.
\item 
We first consider the case $\F_{j'}(t,\R(x))-\F_{j}(t,x)\not\in \text{span}\{Dg_{j,j'} (x)^T\}$ for some $t\geq 0$, $x\in\G_{j,j'}$. Choose $z\in\{z:(\F_{j'}(t,\R(x))-\F_{j}(t,x))^Tz=0 \text{ and }(Dg_{j,j'}(x)) z\not =0\}$. Then
\begin{align}
  \label{eq:8}
  \Xi(t,x) z= z+\beta (\F_{j'}(t,\R(x))-\F_{j}(t,x))
\end{align}
where $\beta=\frac{1}{\dot{g}_{j,j'}(x)}Dg_{j,j'}(x) z$. But then $  \Xi(t,x) z$ is the hypotenuse of a right triangle with legs $z$ and $\beta (\F_{j'}(t,\R(x))-\F_{j}(t,x))$ with nonzero length so that $|\Xi(t,x) z|_2>|z|_2$ and therefore  $\|\Xi(t,x)\|_2>1$.

Thus we have shown that $\|\Xi(t,x)\|_{j,j'}\leq 1$ implies $\F_{j'}(t,\R(x))-\F_{j}(t,x)\in \text{span}\{Dg_{j,j'} (x)^T\}$, \emph{i.e.}, $\eqref{eq:26}$ holds for some $\alpha(t,x)$. We now show that, in particular, \eqref{eq:450} holds and, moreover, \eqref{eq:26}-\eqref{eq:450} implies $\|\Xi(t,x)\|_{j,j'}\leq 1$.

Suppose \eqref{eq:26} holds. Then, for all $t\geq 0$ and all $x\in\G_{j,j'}$,
\begin{align}
  \label{eq:48}
  \Xi(t,x)=I+\frac{\alpha(t,x)}{\dot{g}_{j,j'}(t,x)}Dg_{j,j'} (x)^T Dg_{j,j'}(x)
\end{align}
and 
  \end{enumerate}
\begin{align}
\nonumber &\|\Xi(t,x)\|_2\leq 1\\
  \label{eq:53}&\iff   \|\Xi(t,x)\|^2_2\leq 1\\
&\iff \Xi(t,x)^T \Xi(t,x)\preceq I\\
\nonumber &\iff \Big(I+\left(2 \frac{\alpha(t,x)}{\dot{g}_{j,j'}(t,x)}+\alpha(t,x)^2|Dg_{j,j'}(x)|_2^2\right) \\
&\qquad \qquad \qquad \cdot Dg_{j,j'} (x)^T Dg_{j,j'}(x)\Big)\preceq I\\
&\iff \left(2 \frac{\alpha(t,x)}{\dot{g}_{j,j'}(t,x)}+\alpha(t,x)^2|Dg_{j,j'}(x)|_2^2\right)\leq 0\\
&\iff \eqref{eq:450}\text{ holds.}
\end{align}
\end{proof}

\begin{example}
\newcommand{\jm}{-}
\newcommand{\jp}{+}
 A common special class of hybrid systems for which Proposition \ref{prop:main} is applicable is the class of \emph{piecewise-smooth systems} for which the flow is continuous but is governed by a vector field that is only piecewise-smooth. 
Contraction of such systems has been previously analyzed in \cite{Fiore:2016fj}. 
In this example, we show how Theorem \ref{thm:main} and Proposition \ref{prop:main} apply to such systems and compare to results reported in \cite{Fiore:2016fj}.

Consider a hybrid system with two modes $\J=\{\jm,\jp\}$ with $\D_{\jm}=\D_{\jp}=\mathbb{R}^n$ for some $n$, time-invariant vector field $\F$, and time-invariant guard sets defined as $\G_{\jp,\jm}=\{x:g(x)\leq 0\}$ and $\G_{\jm,\jp}=\{x:-g(x)\leq 0\}$ for some continuously differentiable $g$. Notice that $\D_{{\pm}}\backslash \G=\{x:\pm g(x)>0\}$. Further suppose $\R(x)=x$ for all $x\in\G$, i.e., identity reset map, and consider $|\cdot|_{\jm}=|\cdot|_{\jp}=|\cdot|$ for some norm $|\cdot|$. Here, $\{x:g(x)=0\}$ is the \emph{switching surface} or \emph{switching manifold}. While we study the case with only two modes, the basic idea extends to a state space partitioned into a collection of disjoint open sets separated by codimension-1 switching surfaces.

Suppose the system satisfies the conditions of Theorem \ref{thm:main} and is therefore contractive.  Consider some $x\in\G_{\jm,\jp}$ and note that by Assumption  \ref{assum:flow}.1, we must have $D g(x) \cdot F_{\jp}(x)> 0$ so that the trajectory initialized in $\D_\jm$ at $x$ transitions through the switching surface and flows away from the surface in mode $\jp$, avoiding \emph{sliding}~\rev{\cite{Jeffrey2014-nt}}. 
\rev{(Section~\ref{sec:disc:sliding} discusses how our approach applies in the presence of sliding.)}

Define $\beta(x)=\frac{1}{Dg(x)\cdot \F_{\jm}(x)}$ and $M(x)=(\F_{\jp}(x)-\F_{\jm}(x))\cdot D g(x))$, and observe that $\|\Xi(x)\| = \|I+\beta(x) M(x)\|$. Further, note that  Proposition \ref{prop:main}, part 1  coupled with condition \eqref{eq:thm:Xi} of Theorem~\ref{thm:main} implies that, necessarily, $\|\Xi(x)\|=1$. 

We now show that, also, necessarily $\mu(M(x)) = 0$. To see this, consider  $  \|I+r\beta(x) M(t,x)\|$ for $r\in[0,1]$.  Then
\begin{align}
  \label{eq:6}
  1\leq  \|I+r\beta(x) M(x)\|&=\|(1-r)I+r(I+\beta(x) M(x))\|\\
&\leq \|(1-r)I\|+\|r(I+\beta(x) M(x))\|\\
&=1
\end{align}
where the first inequality holds by the same argument of Proposition \ref{prop:main} part 1. Therefore, $\|I+hM(x)\|=1$ for all sufficiently small $h>0$. Then $\mu(M(x))=\lim_{h\to 0^+}\frac{1}{h}(\|I+hM(x)\|-1)=0$.

Comparing to \cite{Fiore:2016fj}, we see that $\mu(M(x))=0$ along with $\mu(D_x\F_{\pm}(x))\leq c$ for some $c<0$, \emph{i.e.}, \eqref{eq:thm:mu} of Theorem \ref{thm:main}, are sufficient conditions for contraction as given in \cite [Theorem 6]{Fiore:2016fj}. 
Thus, the conditions of Theorem \ref{thm:main}, when specialized to piecewise-smooth systems with no sliding modes, implies the conditions provided in \cite [Theorem 6]{Fiore:2016fj}.
\end{example}

\begin{rem}[summary of main result]
Our main contribution is Theorem~\ref{thm:main}, where we generalize infinitesimal contractivity analysis to the class of hybrid systems satisfying Assumptions~\ref{assum:system} and~\ref{assum:flow}.
This generalization has intuitive appeal, since it combines infinitesimal conditions on continuous-time flow (via the matrix measure of the vector field derivative,~\eqref{eq:thm:mu}) and discrete-time reset (via the induced norm of the saltation matrix,~\eqref{eq:thm:Xi}) that parallel the conditions imposed separately in prior work on smooth continuous-time and discrete-time systems, and establishes contraction with respect to the hybrid system's intrinsic distance function,~\eqref{eq:thm:d}.
With bounds on \emph{dwell time}, i.e. the time between discrete transitions, our tools yield a bound in Corollary~\ref{cor:main} on the intrinsic distance between trajectories regardless of whether the dynamics are contractive or expansive.
Proposition~\ref{prop:main} provides specialized results when the reset is simply a translation, a case that relates to prior work on infinitesimal contraction of nonsmooth vector fields.
We say a hybrid system is \emph{contractive} if~\eqref{eq:cor:zero} holds for all trajectories.
\end{rem}

\section{Applications}
\label{sec:appl}
\renewcommand{\SS}{\texttt{S}}
\newcommand{\CC}{\texttt{C}}

In this section, we study the implications of our results in two application domains.
First, Section~\ref{sec:appl:traffic} considers a hybrid system with identity resets that arises in the study of vehicular traffic flow. 
Traffic congestion can cause a discontinuous change in the speed of traffic flow on a segment of road, 
and the speed does not recover immediately as congestion decreases, so the traffic flow exhibits a hysteresis effect; this hysteresis is important for accurately capturing traffic flow patterns and requires a hybrid model of the dynamics.
Second, Section~\ref{sec:appl:mech} presents a class of hybrid systems with non-identity resets that arise when modeling mechanical systems subject to unilateral constraints.
Several variations are considered that illustrate application of our approach to systems wherein the state dimension and applicable norm changes through reset, 
\rev{and to a system that does not restrict dwell time or mode sequence.
In the interest of readability, we focus on specifying hybrid models and assessing infinitesimal contractivity for these systems, postponing a discussion of why these systems satisfy Assumptions~\ref{assum:system} and~\ref{assum:flow} to Section~\ref{sec:disc:assum}.
}

\subsection{Traffic flow with capacity drop}
\label{sec:appl:traffic}

  Consider a length of freeway divided into two segments or \emph{links}. The state of the system is the traffic \emph{density} on the two links.  Traffic flows from the first segment to the second. The second link has a finite \emph{jam density} $\xjam_2>0$, and we consider link 1 to have infinite capacity so that always the state $x$ satisfies $x\in\mathcal{X}=[0,\infty)\times [0,\xjam_2]\subset \mathbb{R}^2$.

The system has two modes, an \emph{uncongested} (resp., \emph{congested}) mode for which the flow between the two links depends only on the density of the upstream (resp., downstream) link. The dynamics of the uncongested mode is
  \begin{align}
    \dot{x}_1&=u(t)-\Delta_1(x_1)\\
    \dot{x}_2&= \Delta_1(x_1)-\Delta_2(x_2)
  \end{align}
for which we write $\dot{x}=\F_\text{uncon}(x,t)$ assuming a fixed $u(t)$, and for the congested mode is
\begin{align}
  \label{eq:17}
    \dot{x}_1&=u(t)-S_2(x_2)\\
    \dot{x}_2&=S_2(x_2)-\Delta_2(x_2)  
\end{align}
for which we write $\dot{x}=\F_\text{con}(x,t)$ where $\Delta_1$ and $\Delta_2$ are continuously differentiable and strictly increasing \emph{demand} functions satisfying $\Delta_1(0)=\Delta_2(0)=0$, and $S_2$ is a continuously differentiable and strictly decreasing \emph{supply} function satisfying $S_2(\xjam_2)=0$; see \cite{coogan2015compartmental} for further details of the model.

The system is in the congested mode only (but not necessarily) if $  \Delta_1(x_1)\geq  S_2(x_2)$. 
Moreover, empirical studies suggest 
that traffic flow exhibits a hysteresis effect such that traffic remains in the uncongested mode until $x_2\geq \bar{x}_2$ for some $\bar{x}_2$ and does not return to the uncongested mode until $x_2\leq \ulx_2$ for some $\ulx_2<\bar{x}_2$ \cite{Cassidy:1999vl, Laval:2006zp}. Here, we assume $\bar{x}_2\in [0,\xcrit_2)$ where $\xcrit_2$ is the unique density satisfying 
$\Delta_2(\xcrit_2)=S_2(\xcrit_2)$; see Fig.~\ref{fig:traffic2}. This effect is called \emph{capacity drop}.

We model the traffic flow as a hybrid system with four modes $\D_{\SS\CC}, \D_{\bar{\SS}\CC}, \D_{\SS\bar{\CC}},  \D_{\bar{\SS}\bar{\CC}}$ where 
\begin{align}
  \label{eq:49}
  \D_{\SS\CC}&=\mathcal{X}\cap\{x:\Delta_1(x_1)\leq S_2(x_2)\}\cap\{x:x_2\geq \ulx_2\},\\
\D_{\bar{\SS}\CC}&=\mathcal{X}\cap\{x:\Delta_1(x_1)\geq S_2(x_2)\}\cap\{x:x_2\geq \ulx_2\},\\
\D_{\SS\bar{\CC}}&=\mathcal{X}\cap\{x:\Delta_1(x_1)\leq S_2(x_2)\}\cap\{x:x_2\leq \bar{x}_2\},\\
\D_{\bar{\SS}\bar{\CC}}&=\mathcal{X}\cap\{x:\Delta_1(x_1)\geq S_2(x_2)\}\cap\{x:x_2\leq \bar{x}_2\},
\end{align}
and the index set is given by $\mathcal{J}=\{\SS{\CC},\bar{\SS}\CC,\rev{{\SS}\bar{\CC}},\bar{\SS}\bar{\CC}\}$. Furthermore, $\F_{\SS\CC}= \F_{\SS\bar{\CC}}= \F_{\bar{\SS}\bar{\CC}}=\F_\text{uncon}$ and $\F_{\bar{\SS}\CC}=\F_\text{con}$. As a mnemonic, $\SS$ indicates that $\Delta_1(x_1)\leq S_2(x_2)$ so that adequate downstream supply is available, and $\bar{\SS}$ indicates the opposite. Similarly, $\CC$ indicates the status of the hysteresis effect so that the congestion mode is impossible for modes with $\bar{\CC}$.

\begin{figure}
  \centering

   \begin{tikzpicture}[scale=.08]
\fill[gray] (15,30) rectangle (38,38);
\fill[gray] (37.8,32) rectangle (70,38);
\fill[gray] (37.9,30) .. controls +(8,0) and (42,32) .. (50,32) -- (38,38);
\draw[dashed, white] (15,32) -- (40,32);
\draw[dashed, white] (15,34) -- (70,34);
\draw[dashed, white] (15,36) -- (70,36);
  \end{tikzpicture}  

\tikzstyle{link}=[line width=1pt, ->,>=latex, black]
\tikzstyle{junc}=[draw,circle,inner sep=1pt,minimum width=5pt,draw=black]
  \begin{tikzpicture}
\node[junc] (a) at (0,0)  {};
\draw[link,->] (-2,0) -- node[above]{\footnotesize  Link 1} node[below]{\footnotesize $x_1$, density}(a);
\draw[link,->] (a) -- node[above]{\footnotesize  Link 2} node[below]{\footnotesize $x_2$, density} (2,0);
  \end{tikzpicture}
\begin{tabular}{ c c}
\hspace{-10pt}
  \begin{tikzpicture}
\begin{axis}[width=1.8in,
     axis y line=left,
   axis x line=bottom,
every axis x label/.style={at={(rel axis cs:1,0)},anchor=west},
every axis y label/.style={at={(rel axis cs:0,1)},anchor=south},
   xmin=0, xmax=180, ymax=3400,ymin=0, 
   xtick={160},
   xticklabels={$\xjam_1$},
   ytick={\empty},
    xlabel={$x_1$},
    ylabel={flow rate},
    clip=false,
    legend style={at={(.99,.99)},anchor=north east,,cells={align=left}},
    legend cell align=left,
    at={(-260,0)},     ]
    \addplot[name path=A,draw=cfblue,line width=2pt,domain=0:160,samples=50] (\x,{2400*(1-exp(-\x/33))});
    \node[anchor=north] at (160,2400) {\footnotesize $\Delta_1(x_1)$};
\end{axis}
\end{tikzpicture}
&
\hspace{-20pt}
    \begin{tikzpicture}
\begin{axis}[width=1.8in,
     axis y line=left,
   axis x line=bottom,
every axis x label/.style={at={(rel axis cs:1,0)},anchor=west},
every axis y label/.style={at={(rel axis cs:0,1)},anchor=south},
   xmin=0, xmax=180, ymax=3400,ymin=0, 
   xtick={35,65,75,160},
   xticklabels={$\ul{x}_2$,$\bar{x}_2$,, $\xjam_2$},
   ytick={\empty},
    xlabel={$x_2$},
    ylabel={flow rate},
    clip=false,
    legend style={at={(.99,.99)},anchor=north east,,cells={align=left}},
    legend cell align=left,
    at={(-260,0)},     ]
    \addplot[name path=A,draw=cfblue,line width=2pt,domain=0:160,samples=50] (\x,{1900*(1-exp(-\x/33))});
    \addplot[name path=B,draw=dorange,line width=2pt,domain=0:160,samples=2] (\x,{20*(160-x)});
    \path[name path=axis] (axis cs:35,0) -- (axis cs:65,0);
\addplot[fill=dcompb!50]
    fill between[
        of=B and axis,
        soft clip={domain=35:65},
    ];
    \addplot[dashed] coordinates {(75,0) (75,1700)};
    \node[inner sep=0pt] (hyst) at (110,2600) {\footnotesize Hysteresis};
    \draw[->] (hyst.west) to [bend right] +(-25,-650);
    \node[anchor=north] at (160,1900) {\footnotesize $\Delta_2(x_2)$};
    \node[anchor=south] at (160,300) {\footnotesize  $S_2(x_2)$};
    \node[inner sep=0pt, anchor=west] (xcrit) at (90,400) {\footnotesize $\xcrit_2$};
    \draw[->] (xcrit.west) to [bend right] +(-15,0);
\end{axis}
\end{tikzpicture}
\end{tabular}
\caption{Traffic flows from link 1 to link 2. Flow at the interface of link 1 and link 2 depends on the demand $\Delta_1(x_1)$ of link 1 and the supply $S_2(x_2)$ of link 2 and exhibits a hysteresis effect. Traffic exits the network at a flow rate equal to the demand $\Delta_2(x_2)$ of link 2.}
  \label{fig:traffic2}
\end{figure}
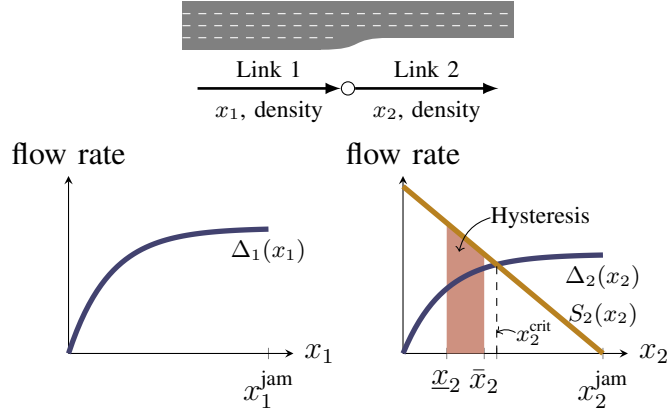

\begin{figure}
  \centering
\begin{tabular}{c c}
\hspace{-12pt}
  \begin{tikzpicture}
  \begin{axis}[width=1.8in,
    height=1.8in,
   axis y line=left,
   axis x line=bottom,
every axis x label/.style={at={(rel axis cs:1,0)},anchor=west},
every axis y label/.style={at={(rel axis cs:0,1)},anchor=south},
   xmin=0, xmax=220, ymax=170,ymin=0, 
   ytick={40,100,160},
   yticklabels={$\ul{x}_2$,$\bar{x}_2$, $\xjam_2$},
   xtick={\empty},
    xlabel={$x_1$},
    ylabel={$x_2$},
    clip=false,
    legend style={at={(.99,.99)},anchor=north east,,cells={align=left}},
    legend cell align=left,
    at={(0,0)}
    ]
    \addplot[name path=A,draw=dorange,line width=2pt,domain=0:200,samples=50] (\x,{-(2200*(1-exp(-\x/33))-3200)/20});
    \addplot[name path=B,draw=cfblue,line width=2pt,domain=0:200,samples=2] (\x,{40});
    \addplot[name path=C,draw=dorange,line width=0pt,domain=0:200,samples=2] (\x,{160});
    \addplot[llblue!60] fill between[of=A and B];
    \addplot[lorange!60] fill between[of=A and C];
    \node[anchor=south] (a) at (130,75) {\footnotesize $\Delta_1(x_1)=S_2(x_2)$};
    \draw [->] ($(a.north west)+(10,-15)$) to [bend right] +(-22,12);
    \node at (100,130) {$\D_{\bar{\SS}\CC}$};
    \node at (25,60) {$\D_{{\SS}\CC}$};
\end{axis}
\end{tikzpicture}
&
\hspace{-20pt}
  \begin{tikzpicture}
  \begin{axis}[width=1.8in,
    height=1.8in,
   axis y line=left,
   axis x line=bottom,
every axis x label/.style={at={(rel axis cs:1,0)},anchor=west},
every axis y label/.style={at={(rel axis cs:0,1)},anchor=south},
   xmin=0, xmax=220, ymax=170,ymin=0, 
   ytick={40,100,160},
   yticklabels={$\ul{x}_2$,$\bar{x}_2$, $\xjam_2$},
   xtick={\empty},
    xlabel={$x_1$},
    ylabel={$x_2$},
    clip=false,
    legend style={at={(.99,.99)},anchor=north east,,cells={align=left}},
    legend cell align=left,
    at={(300,0)},
    ]
    \addplot[name path=A,draw=cfblue,line width=2pt,domain=25:200,samples=50] (\x,{-(2200*(1-exp(-\x/33))-3200)/20});
    \addplot[name path=B,draw=cfblue,line width=0pt,domain=0:200,samples=2] (\x,{00});
    \addplot[name path=C,draw=cfblue,line width=2pt,domain=0:200,samples=2] (\x,{100});
    \addplot[name path=D,draw=cfblue,line width=2pt,domain=0:25,samples=2] (\x,{100});
    \addplot[llblue!60] fill between[of=C and B];
    \addplot[llblue!60] fill between[of=A and C];
    \node[anchor=north] (a) at (85,38) {\footnotesize $\Delta_1(x_1)=S_2(x_2)$};
    \draw [->] ($(a.north east)+(-10,-15)$) to [bend right] +(22,25);
    \node at (150,80) {$\D_{\bar{\SS}\bar{\CC}}$};
    \node at (25,60) {$\D_{{\SS}\bar{\CC}}$};
\end{axis}
\end{tikzpicture}
\end{tabular}
\caption{The traffic network is modeled as a hybrid system with four modes, $\mathcal{J}=\{\SS\bar{\CC},\bar{\SS}\CC,\rev{{\SS}\bar{\CC}},\bar{\SS}\bar{\CC}\}$. The only non-identity \rev{saltation matrix occurs along the transition from $\mathcal{D}_{\bar{\SS}\bar{\CC}}$ to $\mathcal{D}_{\bar{\SS}\CC}$}.} %
  \label{fig:traffic}
\end{figure}

Define the guard functions
\begin{align}
  g_{\SS\CC,\bar{\SS}\CC}(x)&=  g_{\SS\bar{\CC},\bar{\SS}\bar{\CC}}(x)=S_2(x_2)-\Delta_1(x_1),\\
  g_{\bar{\SS}\CC,\SS\CC}(x)&=  g_{\bar{\SS}\bar{\CC},\SS\bar{\CC}}(x)=\Delta_1(x_1)-S_2(x_2),\\
g_{\SS\bar{\CC},\SS\CC}(x)&=g_{\bar{\SS}\bar{\CC},\bar{\SS}\CC}(x)=\bar{x}_2 -x_2,\\
g_{\SS\CC,\SS\bar{\CC}}(x)&=x_2-\ulx_2.
\end{align}

If no guard function is specified between two modes, then no transition is possible between those modes. \rev{In all cases, the reset map is the identity map.} For all $j,j'\in \mathcal{J}$ such that $g_{j,j'}$ is defined, let $\G_{j,j'}=\{x:g_{j,j'}(x)\leq 0\}\cap \D_j$, and let $\G_j=\cup_{j'\in \J}\G_{j,j'}$ for each $j\in\J$.

We have that
\begin{align}
  \label{eq:28}
  J_\text{uncon}(x)&=D_x\F_\text{uncon}(x,t)=
  \begin{bmatrix}
    -D\Delta_1(x_1)&0\\
D\Delta_1(x_1)&-D\Delta_2(x_2)
  \end{bmatrix},\\
  J_\text{con}(x)&=D_x\F_\text{con}(x,t)=
                     \begin{bmatrix}
                       0&-DS_2(x_2)\\
                       0&DS_2(x_2)-D\Delta_2(x_2)
                     \end{bmatrix}.
\end{align}

Let $|\cdot|_1$ be the standard one-norm and $\mu_1$ the corresponding matrix measure. It can be verified that
\begin{align}
  \label{eq:30}
\mu_1(  J_\text{uncon}(x))&\leq 0\ \forall x\in\mathcal{X},\quad \text{and}\\
\quad \mu_1(J_\text{con}(x))&\leq 0\ \forall x\in\mathcal{X}.
\end{align}
Now consider a trajectory in mode $\D_{\bar{\SS}\bar{\CC}}$ transitioning to $\D_{\bar{\SS}{\CC}}$ so that $S_2(x_2)\leq \Delta_1(x_1)$ and the system experiences a capacity drop so that the dynamics transition from uncongested to congested. Computing the saltation matrix $\Xi$ for $x$ such that $g_{{\bar{\SS}\bar{\CC}}, {\bar{\SS}{\CC}}}(x)=0$, we have
\begin{align}
  \label{eq:traffic:salt}
 & \Xi_{{\bar{\SS}\bar{\CC}}, {\bar{\SS}{\CC}}}(t,x)=I+\frac{(\F_\text{con}(t,x)-\F_\text{uncon}(t,x))\cdot D_x g_{{\bar{\SS}\bar{\CC}}, {\bar{\SS}{\CC}}} (x)}{D_x g_{{\bar{\SS}\bar{\CC}}, {\bar{\SS}{\CC}}} (x)\cdot \F_\text{uncon}(t,x)}\\
&=I+  \frac{-1}{\Delta_1(x_1)-\Delta_2(\bar{x}_2)}\begin{bmatrix}
    \Delta_1(x_1)-S_2(\bar{x}_2)\\
S_2(\bar{x}_2)-    \Delta_1(x_1)
  \end{bmatrix}\cdot
  \begin{bmatrix}
    0&-1
  \end{bmatrix}
\end{align}
for all $x\in \{x:x_2=\bar{x}_2\}=\G_{{\bar{\SS}\bar{\CC}}, {\bar{\SS}{\CC}}}(x)$. Let $ \rho(x_1)=\frac{\Delta_1(x_1) -S_2(\bar{x}_2)}{\Delta_1(x_1)-\Delta_2(\bar{x}_2)}$ so that
\begin{align}
  \label{eq:52}
    \Xi_{{\bar{\SS}\bar{\CC}}, {\bar{\SS}{\CC}}}(t,x)=
  \begin{bmatrix}
    1&\rho(x_1)\\
    0&1-\rho(x_1)
  \end{bmatrix}
\end{align}
for all $x\in \{x:x_2=\bar{x}_2\}$. Because $\bar{x}_2<\xcrit_2$, it holds that $\Delta_2(\bar{x}_2)<S_2(\bar{x}_2)$ and therefore
\begin{align}
  \label{eq:47}
0\leq  \rho(x_1)<1\quad \forall x_1\in\{x_1:  \Delta_1(x_1) \geq S_2(\bar{x}_2)\}.
\end{align}
Therefore, $\|    \Xi_{{\bar{\SS}\bar{\CC}}, {\bar{\SS}{\CC}}}(t,x)\|_1=1$ for all $x\in \{x:x_2=\bar{x}_2\}=\G_{{\bar{\SS}\bar{\CC}}, {\bar{\SS}{\CC}}}$.

For all $(j,j')\neq ({{\bar{\SS}\bar{\CC}}, {\bar{\SS}{\CC}}})$ such that $\G_{j,j'}$ is nonempty, it can be verified that $\F_{j'}(x)=\F_j(x)$ for all $x\in \G_{j,j'}$ so that $\Xi_{j,j'}(t,x)=I$ and trivially $\|\Xi_{j,j'}(t,x)\|_1=1$. Applying Theorem \ref{thm:main}, we conclude that
\begin{align}
  \label{eq:25}
  |y(t)-x(t)|_1\leq |y(0)-x(0)|_1
\end{align}
for any pair of trajectories $x(t), y(t)$ of the traffic flow system with initial conditions $y(0),x(0)$ subject to any input $u(t)$, that is, the system is nonexpansive.

\rev{
\subsection{Mechanical systems subject to unilateral constraints}
\label{sec:appl:mech}

Consider a \emph{mechanical} system with $d\in\Nbb$ degrees-of-freedom (DOF) subject to $n\in\Nbb$ \emph{unilateral} constraints~\cite{Johnson2016-nh},
\begin{subequations}\label{eq:ddq}
\begin{align}
  M(q)\ddot{q} & = 
  f(t,q,\dot{q})\ \text{s.t.}\ a(q) \ge 0,
  \label{eq:ddq:cont}\\
  \dot{q}^+ & = \Delta(q)\dot{q}^-,\ a(q) = 0,
  \label{eq:ddq:disc}
\end{align}
\end{subequations}
where:
$q\in Q = \Rbb^d$ are generalized configuration coordinates;
$M(q)\in\Rbb^{d\times d}$ is the mass matrix;
$f(t,q,\dot{q})\in\Rbb^d$ is the time-varying vector of Coriolis, potential, and applied forces;
$a(q)\in\Rbb^n$ is the constraint function;
and
we regard the inequality $a(q) \ge 0$ as being enforced componentwise.
In what follows we will restrict our attention to linear ``MCK'' systems subject to linear unilateral constraints%
\footnote{If the constraint function is affine, $a(q) = A\cdot(q - q_0)$, we translate coordinates as $\ol{q} = q - q_0$; 
this shift offsets the potential energy by an affine term $q^\tr K q_0 + \frac{1}{2}q_0^\tr K q_0$; assuming the additional constant force generated by this offset $-K q_0$ lies in the span of the columns of $B$ so $-K q_0 = B \mu$ for some $\mu\in\mathbb{R}^m$, we could incorporate it into the open--loop input, $\ol{u}(t) = u(t) + \mu$.}%
,
\eqnn{\label{eq:appl:mech:viscodn}
M \ddot{q} + C \dot{q} + K q = B u(t) \ \text{s.t.}\ a(q) = A q \ge 0,
}
and restrict to impact with the constraints that is (\emph{perfectly}) \emph{plastic},
where:
$M = M^\top > 0$,
$C = C^\top \geq 0$,
and
$K = K^\top > 0$
are constant (inertia, damping, and elasticity) $d\times d$ matrices,
$B \in \Rbb^{d\times m}$ is a constant matrix that maps inputs $u\in\Rbb^m$ to forces,
and
$A \in \Rbb^{n\times d}$ is a constant matrix that determines the constraints.
To satisfy Assumption~{\ref{assum:flow}.2} (continuity of hybrid system flow) in this class of systems, it is necessary~\cite[Thm.~20]{Ballard2000-ui} that the constraint normals are orthogonal%
\footnote{Strictly speaking,~\eqref{eq:appl:mech:orth} specifies that the constraints are ortho\emph{normal}, whereas orthogonality requires only that $A M^{-1} A^\top$ is diagonal; without loss of generality, we assume coordinates have been chosen so that the diagonal matrix is the identity to simplify the (already rather complex) formulas in what follows.}%
,
\eqnn{
\label{eq:appl:mech:orth}
A\ M^{-1}  A^\top = I_{n\times n};
}
this algebraic fact will be leveraged repeatedly in what follows.
Given a subset of constraint indices $J\subset\set{1,\dots,n}$,
we will let 
$n_J = \card{J}$ denote the number of elements in $J$ and
$A_J\in\Rbb^{n_J\times d}$ denote the matrix obtained by extracting rows indexed by $J$ from $A$; 
note that~\eqref{eq:appl:mech:orth} implies that 
$A_J\ M^{-1} A_J^\top = I_{n_J\times n_J}$. 
When the
(perfect, holonomic, scleronomic)%
\footnote{A constraint is:
\emph{perfect} if it only generates force in the direction normal to the constraint surface; 
\emph{holonomic} if it varies with configuration but not velocity;
\emph{scleronomic} if it does not vary with time.
}
constraints $J\subset\set{1,\dots,n}$ activate
(so $a_J(q) = A_J q = 0$),
they apply impulses
$\dot{q}^+ = \Delta_J \dot{q}^-$
where
\eqnn{
\label{eq:appl:mech:impulse}
\Delta_J 
= I - M^{-1}  A_J^\top  (A_J M^{-1}  A_J^\top)^{-1}  A_J
= I - M^{-1} A_J^\top A_J
}
and forces
$A_J^\top\lambda_J(t,q,\dot{q})$
where
\eqnn{
\label{eq:appl:mech:forces}
\lambda_J(t,q,\dot{q}) 
= - (A_J M^{-1}  A_J^\top)^{-1}  A_J  M^{-1} (B u(t) - C \dot{q} - K q) 
= A_J M^{-1} (B u(t) - C \dot{q} - K q).
}
The sum of potential and kinetic energy,
\eqnn{
e(q,\dot{q}) = \frac{1}{2} q^\tr K q + \frac{1}{2} \dot{q}^\tr M \dot{q},
}
can be used to determine a weighting matrix for a $2$--norm%
\footnote{That is, $\norm{x} = \sqrt{\frac{1}{2} x^\tr E\, x}$.}
\eqnn{
\label{eq:appl:mech:E}
E = D^2 e = \mat{cc}{K & 0 \\ 0 & M}.
}
This norm is attractive for infinitesimal contraction analysis of~\eqref{eq:ddq:cont} since,
with $x = (q,\dot{q})$ denoting the system state and $\dot{x} = \F_\emptyset(t,x)$ denoting the continuous-time dynamics in the unconstrained mode ($J = \emptyset$), 
and recalling that the matrix measure for a 2-norm determined by weighting matrix $E$ is given by
\eqnn{
\mu(X) = \max\spec\frac{1}{2}\paren{X^\tr \cdot E + E \cdot X},
}
we find that using the energy--induced $2$--norm yields a negative semidefinite matrix
\eqnn{
\label{eq:appl:mech:contract}
\frac{1}{2}\paren{D_x \F_\emptyset^\tr\cdot E + E \cdot D_x \F_\emptyset} & 
= \mat{cc}{0 & 0 \\ 0 & -C} \leq 0,
}
whence the unconstrained flow is non-expansive overall and infinitesimally contractive in velocity coordinates affected by damping.
The collection of constraint modes $\set{J\subset\set{1,\dots,n}}$ with corresponding vector fields and velocity resets together determine a hybrid system in the framework from Section~\ref{sec:prelim}; we refer the interested reader to~\cite{Johnson2016-nh} for more details.

Note that the sum of potential and kinetic energy decreases monotonically in~\eqref{eq:ddq}
since%
~\eqref{eq:appl:mech:contract} implies that $\dot{e} \leq 0$ in~\eqref{eq:ddq:cont}
and%
~\eqref{eq:appl:mech:impulse} implies that $e^+ = e(q,\dot{q}^+) \leq e(q,\dot{q}^-) = e^-$ in~\eqref{eq:ddq:disc}. 
These observations lead us to intuit that the hybrid system specified by~\eqref{eq:ddq} ought to be (infinitesimally) contractive (at least, non-expansive).
The sequence of examples in the remainder of this section illustrate cases where this intuition holds and cases where it is violated.
Specifically, we assess infinitesimal contractivity in simple systems: 
with few (1, 2, or 3) degrees-of-freedom (DOF), 
\emph{hard} or \emph{soft} constraints, 
and 
\emph{elastic} or \emph{viscoelastic} spring-dampers, 
to highlight obstacles to infinitesimal contractivity in~\eqref{eq:ddq}:
\begin{enumerate}[label=\ref{sec:appl:mech}\arabic*,leftmargin=0.4in]
  \item (Fig.~\ref{fig:mech:1dof}) is a \emph{linear impact oscillator}~\cite{Nordmark1991-ui} with $1$ DOF ($Q = \Rbb^1$), $1$ hard constraint ($a:Q\into\Rbb^1$), and elastic spring-dampers where continuous-time flow and discrete-time reset are contractive;
  \item (Fig.~\ref{fig:mech:2dof}) extends~\ref{sec:appl:mech:1dof} to $2$ DOF ($Q = \Rbb^2$), 
  finding that continuous-time flow is contractive
  but 
  discrete-time reset is expansive when constraints activate;
  \item (Fig.~\ref{fig:mech:soft}) modifies~\ref{sec:appl:mech:1dof} by relaxing the hard constraint using the penalty method from~\cite{Tornambe1999-dr},
  finding that 
  continuous-time flow is contractive
  but
  discrete-time reset is expansive when constraints \emph{de}activate;
\item (Fig.~\ref{fig:mech:visco}) modifies~\ref{sec:appl:mech:2dof} by employing visco-elastic spring-dampers, 
  finding that continuous-time flow is 
  contractive %
  and 
  discrete-time reset is non-expansive. 
  \item (Fig.~\ref{fig:mech:visco32}) extends~\ref{sec:appl:mech:visco} to $3$ DOF ($Q = \Rbb^3$) and $2$ constraints ($a:Q\into\Rbb^2$), demonstrating application of our approach without restricting mode sequence or dwell time;
  \item extends~\ref{sec:appl:mech:visco32} to $d$ DOF ($Q = \Rbb^d$) and $n$ constraints ($a:Q\into\Rbb^n$), demonstrating application of our approach to a system with an arbitrary number of degrees-of-freedom and constraints.
\end{enumerate}
\noindent
These mechanical systems are
depicted
in Fig.~\ref{fig:mech}.
Note that 
$\C = \set{(q,\dot{q}) : a(q) \ge 0}\subset TQ$ is forward-invariant under~\eqref{eq:ddq} for all but~\ref{sec:appl:mech:soft}, where $\C = TQ$ is forward-invariant;
these subsets will be used as the \emph{contraction regions} in each application of our results in the remainder of this section.

} %

\begin{figure}
\begin{minipage}[b]{\columnwidth}
  {
  \def\svgwidth{0.9\columnwidth} 
  \resizebox{7cm}{!}{\begingroup%
  \makeatletter%
  \providecommand\color[2][]{%
    \errmessage{(Inkscape) Color is used for the text in Inkscape, but the package 'color.sty' is not loaded}%
    \renewcommand\color[2][]{}%
  }%
  \providecommand\transparent[1]{%
    \errmessage{(Inkscape) Transparency is used (non-zero) for the text in Inkscape, but the package 'transparent.sty' is not loaded}%
    \renewcommand\transparent[1]{}%
  }%
  \providecommand\rotatebox[2]{#2}%
  \newcommand*\fsize{\dimexpr\f@size pt\relax}%
  \newcommand*\lineheight[1]{\fontsize{\fsize}{#1\fsize}\selectfont}%
  \ifx\svgwidth\undefined%
    \setlength{\unitlength}{534.02862549bp}%
    \ifx\svgscale\undefined%
      \relax%
    \else%
      \setlength{\unitlength}{\unitlength * \real{\svgscale}}%
    \fi%
  \else%
    \setlength{\unitlength}{\svgwidth}%
  \fi%
  \global\let\svgwidth\undefined%
  \global\let\svgscale\undefined%
  \makeatother%
  \begin{picture}(1,0.22753887)%
    \lineheight{1}%
    \setlength\tabcolsep{0pt}%
    \put(0,0){\includegraphics[width=\unitlength,page=1]{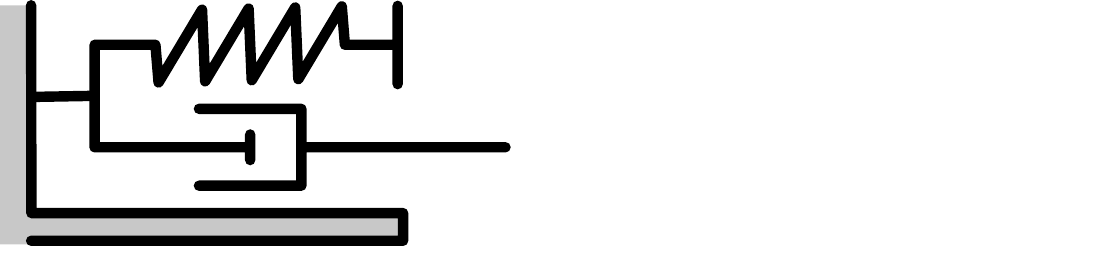}}%
    \put(0.11173893,0.20643902){\color[rgb]{0,0,0}\makebox(0,0)[t]{\lineheight{0}\smash{\begin{tabular}[t]{c}$\kappa$\end{tabular}}}}%
    \put(0.13982731,0.05326364){\color[rgb]{0,0,0}\makebox(0,0)[t]{\lineheight{0}\smash{\begin{tabular}[t]{c}$\beta$\end{tabular}}}}%
    \put(0.50074733,0.08985233){\color[rgb]{0,0,0}\makebox(0,0)[t]{\lineheight{0}\smash{\begin{tabular}[t]{c}$m$\end{tabular}}}}%
    \put(0,0){\includegraphics[width=\unitlength,page=2]{1dof.pdf}}%
    \put(0.40525891,0.1275075){\color[rgb]{0,0,0}\makebox(0,0)[t]{\lineheight{0}\smash{\begin{tabular}[t]{c}$u(t)$\end{tabular}}}}%
  \end{picture}%
\endgroup%
}
  }
  \subcaption{Sec.~\ref{sec:appl:mech:1dof}: 1 DOF, 1 hard constraint, elastic; contractive}\label{fig:mech:1dof}
\end{minipage}%
\\

\begin{minipage}[b]{\columnwidth}
  {
  \def\svgwidth{0.9\columnwidth} 
  \resizebox{!}{1.65cm}{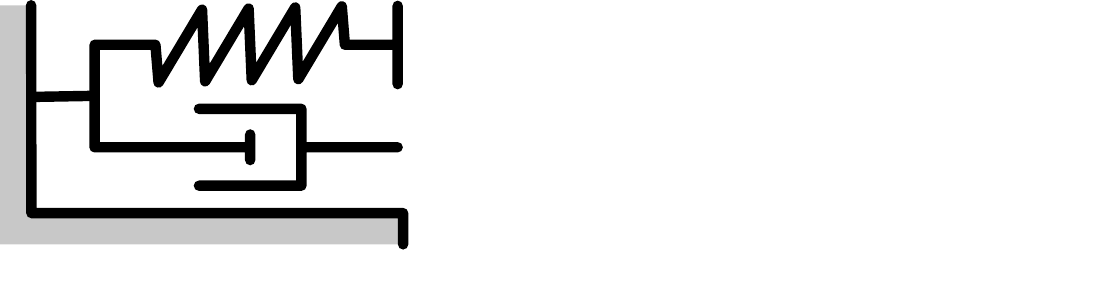}
  }
  \subcaption{Sec.~\ref{sec:appl:mech:2dof}: 2 DOF, 1 hard constraint, elastic; expansive constraint activation}\label{fig:mech:2dof}
\end{minipage}%
\\

\begin{minipage}[b]{\columnwidth}
  {
  \def\svgwidth{0.9\columnwidth} 
  \resizebox{!}{1.65cm}{\begingroup%
  \makeatletter%
  \providecommand\color[2][]{%
    \errmessage{(Inkscape) Color is used for the text in Inkscape, but the package 'color.sty' is not loaded}%
    \renewcommand\color[2][]{}%
  }%
  \providecommand\transparent[1]{%
    \errmessage{(Inkscape) Transparency is used (non-zero) for the text in Inkscape, but the package 'transparent.sty' is not loaded}%
    \renewcommand\transparent[1]{}%
  }%
  \providecommand\rotatebox[2]{#2}%
  \newcommand*\fsize{\dimexpr\f@size pt\relax}%
  \newcommand*\lineheight[1]{\fontsize{\fsize}{#1\fsize}\selectfont}%
  \ifx\svgwidth\undefined%
    \setlength{\unitlength}{534.02862549bp}%
    \ifx\svgscale\undefined%
      \relax%
    \else%
      \setlength{\unitlength}{\unitlength * \real{\svgscale}}%
    \fi%
  \else%
    \setlength{\unitlength}{\svgwidth}%
  \fi%
  \global\let\svgwidth\undefined%
  \global\let\svgscale\undefined%
  \makeatother%
  \begin{picture}(1,0.26135796)%
    \lineheight{1}%
    \setlength\tabcolsep{0pt}%
    \put(0.33773271,0.05676306){\color[rgb]{0,0,0}\makebox(0,0)[t]{\lineheight{0}\smash{\begin{tabular}[t]{c}$\kappa''$\end{tabular}}}}%
    \put(0,0){\includegraphics[width=\unitlength,page=1]{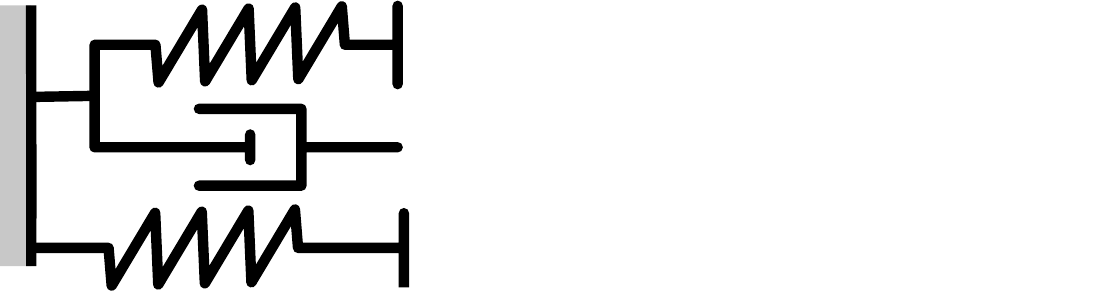}}%
    \put(0.11173894,0.24025835){\color[rgb]{0,0,0}\makebox(0,0)[t]{\lineheight{0}\smash{\begin{tabular}[t]{c}$\kappa$\end{tabular}}}}%
    \put(0.13982732,0.08708297){\color[rgb]{0,0,0}\makebox(0,0)[t]{\lineheight{0}\smash{\begin{tabular}[t]{c}$\beta$\end{tabular}}}}%
    \put(0.50074735,0.12367165){\color[rgb]{0,0,0}\makebox(0,0)[t]{\lineheight{0}\smash{\begin{tabular}[t]{c}$m$\end{tabular}}}}%
    \put(0,0){\includegraphics[width=\unitlength,page=2]{soft.pdf}}%
    \put(0.40525893,0.16694451){\color[rgb]{0,0,0}\makebox(0,0)[t]{\lineheight{0}\smash{\begin{tabular}[t]{c}$u(t)$\end{tabular}}}}%
  \end{picture}%
\endgroup%
}
  }
  \subcaption{Sec.~\ref{sec:appl:mech:soft}: 1 DOF, 1 soft constraint, elastic; expansive constraint deactivation}\label{fig:mech:soft}
\end{minipage}%
\\

\begin{minipage}[b]{\columnwidth}%
  \vspace{.2cm}
  {%
  \def\svgwidth{0.9\columnwidth}%
  \resizebox{7cm}{!}{\input{visco.pdf_tex}}%
  }%
  \subcaption{Sec.~\ref{sec:appl:mech:visco}: 2 DOF, 1 hard constraint, viscoelastic; contractive}\label{fig:mech:visco}
\end{minipage}%
\\
\begin{minipage}[b]{\columnwidth}
  \vspace{.2cm}
  {
  \def\svgwidth{1.0\columnwidth} 
  \resizebox{7cm}{!}{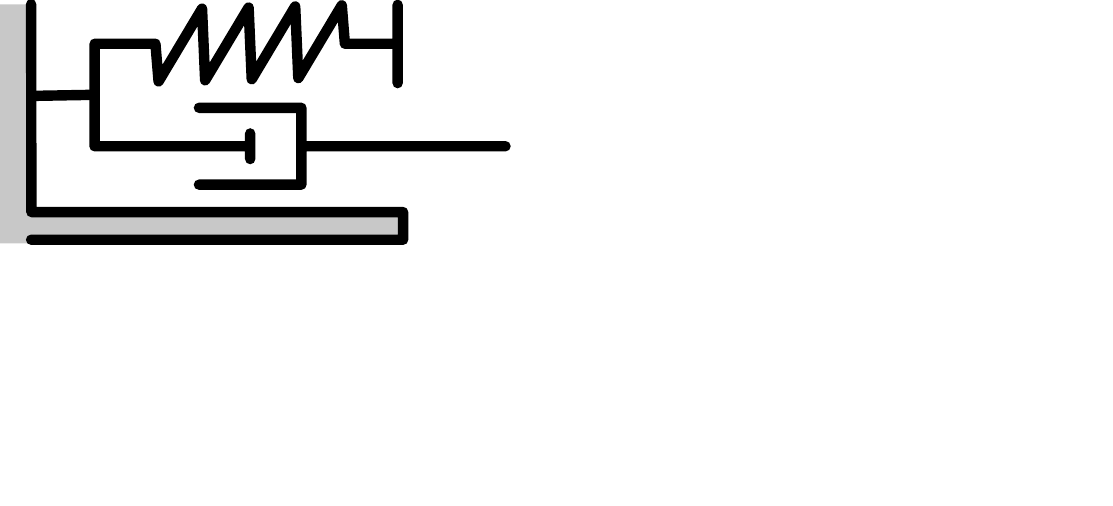}
  }
  \subcaption{Sec.~\ref{sec:appl:mech:visco32}: 3 DOF, 2 hard constraints, viscoelastic; contractive}\label{fig:mech:visco32}
\end{minipage}%
\caption{%
Mechanical systems subject to unilateral constraints considered in Sec.~\ref{sec:appl:mech}.%
}
\label{fig:mech}
\end{figure}

\subsubsection{Elastic spring-damper with $1$ DOF, $1$ constraint}
\label{sec:appl:mech:1dof}
Point mass $m$ moves along a frictionless rail to the right of a mechanical stop positioned at the origin,
impacting plastically if it reaches the stop with negative velocity as illustrated in Fig.~\ref{fig:mech:1dof}.
The mass
is connected to the stop with a parallel spring-damper: 
the viscosity of the damper is denoted $\beta$; 
the stiffness of the spring is denoted $\kappa$;
the spring's rest length is adjusted by a time-varying input $u(t)$.
This mechanical system subject to a unilateral constraint can be modeled in the hybrid systems framework from Section~\ref{sec:prelim} with:
\begin{itemize}
\item[$\D$] $= 
\D_{\uncon}\coprod
\D_{\con}$
where 
$\D_{\uncon} = \set{(q,\dot{q}) \in \Rbb^2}$ 
is the set of continuous states wherein 
the mass is unconstrained
and
$\D_{\con} = \set{0\in \Rbb^0}$ 
is the (singleton) set of continuous states wherein 
the mass is constrained;

\item[$\F$] $: [0,\infty)\times\D\into T\D$ 
defined by 
$\F_{\uncon}(t,q,\dot{q}) = (\dot{q},\ddot{q})$
and
$\F_{\con}(t,0) = 0$,
where
$m \ddot{q} =  
\kappa \left(u(t) - q\right)
- \beta \dot{q};$

\item[$\G$] $= \G_{\TD}\coprod\G_{\LO}$ where 
\eqnn{\label{eq:TD}
\G_{\TD} = \set{(t,q,\dot{q})\in[0,\infty)\times\D_\uncon : q \le 0,\ \dot{q} < 0}
}
is the set of 
states where the mass impacts the stop 
and
\eqnn{\label{eq:LO}
\G_{\LO} = \set{(t,0)\in[0,\infty)\times\D_\con : u(t) \ge 0}
}
is the set of 
states where the mass accelerates off the stop;
 
\item[$\R$] $:\G\into\D$ defined by 
$\R|_{\G_{\TD}}(t,q,\dot{q}) = 0\in\Rbb^0$ (i.e. the position and velocity of the mass are ``forgotten'' since they are both equal to zero following plastic impact with the stop) and
$\R|_{\G_{\LO}}(t,0) = (0,0)\in\Rbb^2$ (i.e. the position and velocity of the mass are reinitialized at zero when the mass accelerates off the stop).
\end{itemize}
\noindent
Note that it was parsimonious but not necessary to remove the mass's position and velocity from the continuous state in the constrained mode; 
we elect to remove these states here and in what follows primarily to illustrate application of this paper's theoretical results in systems wherein the dimension of the continuous state varies in different modes.

Letting the sum of the potential and kinetic energy 
$e = \frac{1}{2} \kappa (u(t) - q)^2 + \frac{1}{2} m \dot{q}^2$ 
determine the weighting matrix for a $2$-norm 
as in~\eqref{eq:appl:mech:E},
we find that
\eqnn{
\frac{1}{2}\paren{D_x \F_\uncon^\tr \cdot E + E \cdot D_x \F_\uncon} = \mat{cc}{0 & 0 \\ 0 & -\beta},
}
whence the unconstrained flow is non-expansive overall and contractive in velocity.
Since $D_x \F_{\con} = 0$, the constrained flow is non-expansive.
The saltation matrices for constraint activation and deactivation are both zero operators, hence their induced norms are zero,
whence discrete-time reset is contractive.
Theorem~\ref{thm:main} implies that the distance between trajectories does not increase over time.
Corollary~\ref{cor:main} yields the intuitive conclusion that the distance between any two trajectories is zero if both trajectories have undergone at least one discrete transition.

\subsubsection{Elastic spring--dampers with $2$ DOF, $1$ hard constraint}
\label{sec:appl:mech:2dof}
Two point masses $m$, $m'$ move along a frictionless rail to the right of a stop positioned at the origin as illustrated in Fig.~\ref{fig:mech:2dof};
mass $m$ impacts plastically if it reaches the stop with negative velocity.
Parallel linear spring--dampers connect the masses to one another, 
and mass $m$ to the stop;
the viscosity of the damper connecting $m$ to the stop is denoted
$\beta$ (resp. $\beta'$ for the damper connecting the two masses);
stiffness of the spring connecting $m$ to the stop is denoted by 
$\kappa$ 
(resp. 
$\kappa'$);
and 
the rest length of this spring is adjusted by a time--varying input $u(t)$.
This mechanical system subject to a unilateral constraint can be modeled in the hybrid systems framework from Section~\ref{sec:prelim} with:
\begin{itemize}
\item[$\D$] $=  \D_{\uncon}\coprod \D_{\lft}$ 
where 
$\D_{\uncon} = \set{(q,q',\dot{q},\dot{q}') \in[0,\infty)\times\Rbb\times\Rbb\times\Rbb}$ 
is the set of continuous states wherein 
mass $m$ is unconstrained
and
$\D_{\lft} = \set{(q',\dot{q}') \in \Rbb\times\Rbb}$ 
is the set of continuous states wherein 
mass $m$ is constrained;

\item[$\F$] $: [0,\infty)\times\D\into\D$ 
defined by 
$\F_{\uncon}(t,q,q',\dot{q},\dot{q}') = (\dot{q},\dot{q}',\ddot{q},\ddot{q}')$
and
$\F_{\lft}(t,q',\dot{q}') = (\dot{q}',\ddot{q}')$
where
\eqnn{
\ddot{q} & = \frac{1}{m} 
\left(
+\kappa \left(q' - q\right)
+ \beta \left(\dot{q}' - \dot{q}\right) 
+ \kappa (u(t) - q) - \beta\dot{q}
\right), \\
\ddot{q}' & = \frac{1}{m'} 
\left(
-\kappa \left(q' - q\right)
- \beta \left(\dot{q}' - \dot{q}\right)\right);
}

\item[$\G$] $= \G_{\TD}\coprod\G_{\LO}$ 
where 
\eqnn{\G_{\TD} = \set{(t,q,q',\dot{q},\dot{q}')\in[0,\infty)\times\D_\uncon : q \le 0,\ \dot{q} < 0}}
is the set of \emph{touchdown} (TD) states wherein mass $m$ impacts the stop 
and
\eqnn{\G_{\LO} = \set{(t,q',\dot{q}')\in[0,\infty)\times\D : \ddot{q} \ge 0}}
is the set of \emph{liftoff} (LO) states wherein mass $m'$ accelerates off the stop;
 
\item[$\R$] $:\G\into\D$ defined by 
$\R|_{\G_{\TD}}(t,q,q',\dot{q},\dot{q}') = (q',\dot{q}')$ (i.e. the position and velocity of mass $m$ are ``forgotten'' since they are both equal to zero when the stop constraint activates) and
$\R|_{\G_{\LO}}(t,q',\dot{q}') = (0,q',0,\dot{q}')$ (i.e. the position and velocity of mass $m$ are reinitialized at zero when stop constraint deactivates).

\end{itemize}

As with the previous example, we let the sum of potential and kinetic energy determine the weighting matrices for $2$--norms that will be used to assess infinitesimal contractivity of continuous--time flow and discrete--time reset.
In the unconstrained mode, the energy is
\eqnn{
e_{\uncon}(q,\dot{q}) &=
\frac{1}{2} \kappa (u(t) - q)^{2} 
+ 
\frac{1}{2} \kappa' \left(q' - q\right)^{2} 
+ 
\frac{1}{2} m \dot{q}^{2}  
+ 
\frac{1}{2} m' \dot{q}'^{2},
}
yielding the metric
\eqnn{
E_{\uncon} = D^2 e_{\uncon} = 
E_{\uncon} &=\left[\begin{matrix}\kappa + \kappa' & - \kappa' & 0 & 0\\- \kappa' & \kappa & 0 & 0\\0 & 0 & m & 0\\0 & 0 & 0 & m'\end{matrix}\right].
}
In the constrained mode, the energy simplifies to 
\eqnn{
e_{\lft}(q,\dot{q}) &=
\frac{1}{2} \kappa u(t)^{2}
+ 
\frac{1}{2} \kappa' q'^{2} 
+ 
\frac{1}{2} m'\dot{q}'^{2},
}
yielding the metric
\eqnn{
E_{\lft} =  D^2 e_{\lft} &=\left[\begin{matrix}\kappa' & 0\\0 & m'\end{matrix}\right].
}

We first consider infinitesimal contractivity of continuous--time flow.
Letting $x = (q,q',\dot{q},\dot{q}')$ denote the continuous state vector in the unconstrained mode and $\dot{x} = \F_{\uncon}(t,x)$ denote its time derivative yields
\eqnn{
\frac{1}{2}\paren{ D_x \F_{\uncon}^\tr \cdot E_{\uncon} + E_{\uncon} \cdot D_x \F_{\uncon} } &=
\diag\paren{
\left[\begin{matrix}0 & 0 \\0 & 0 \end{matrix}\right]
,
\left[\begin{matrix} - (\beta + \beta') & \beta'\\\beta' & - \beta'\end{matrix}\right]},
}
Similarly, letting $x = (q',\dot{q}')$ denote the continuous state vector in the constrained mode and $\dot{x} = \F_{\lft}(t,x)$ denote its time derivative yields
\eqnn{
\frac{1}{2}\paren{D_x \F_{\lft}^\tr \cdot E_{\lft} + E_{\lft} \cdot D_x \F_{\lft}} &=\diag\paren{0,-\beta'}.
}
The spectrum of 
\begin{equation}\frac{1}{2}\paren{D_x \F_{\uncon}^\tr \cdot E_{\uncon} + E_{\uncon} \cdot D_x \F_{\uncon}}\end{equation}
is 
$\set{0,-\frac{1}{2}\paren{\beta+2\beta'\pm\sqrt{\beta^2+4\beta'^2}}}$
and
that of
\begin{equation}\frac{1}{2}\paren{D_x \F_{\lft}^\tr \cdot E_{\lft} + E_{\lft} \cdot D_x \F_{\lft}}\end{equation}
is 
$\set{0,-\beta'}$, 
so the matrix measures of $D_x\F_\uncon$ and $D_x\F_\lft$ are both equal to 0 (zero),
whence continuous--time flow is non--expansive%
\footnote{Since some of the eigenvalues are negative, the flow is actually
\emph{semi--}~\cite[Sec.~2]{Lohmiller1998-xj}
or 
\emph{horizontally}~\cite[Sec.~VII]{Forni2014-pq}
contractive.}
in both the constrained and unconstrained mode.

We now consider infinitesimal contractivity of discrete--time reset, i.e. we evaluate induced norms of saltation matrices.
The \emph{liftoff} (\LO) saltation matrix is
\eqnn{
\Xi_{\LO} &=\left[\begin{matrix}0 & 0\\1 & 0\\0 & 0\\0 & 1\end{matrix}\right];
}
since this matrix is an energy--preserving embedding, it is an isometry with respect to the $2$--norms determined by the energy metric, hence its induced norm is equal to unity.
The \emph{touchdown} (\TD) saltation matrix is
\eqn{
\Xi_{\TD} &=\left[\begin{matrix}0 & 1 & 0 & 0\\- \frac{1}{m'}\beta' & 0 & 0 & 1 \end{matrix}\right].
}
We seek to evaluate the matrix norm $\norm{\Xi_{TD}}$ induced by the vector 2--norms determined by $E_{\uncon}$ in the unconstrained mode and $E_{\lft}$ in the constrained mode:
\eqn{
\norm{\Xi_{\TD}} = \max\set{\norm{\Xi_{\TD} x}_{\lft} : \norm{x}_{\uncon} = 1}.
}
Unfortunately, although an expression for this induced norm is readily obtained using symbolic computer algebra, we were unable to analytically determine when this expression is larger than unity.
However, a straightforward calculation shows that the induced norm is larger than unity for all $\beta'$ sufficiently large: 
noting that the vector 
$v = \paren{\frac{\sqrt{2}}{\sqrt{\kappa+\kappa'}},0,0,0}^\tr\in\D_\uncon$ 
has norm
$\norm{v}_\uncon = 1$
and that the vector
$w = \Xi_{\TD} v = \paren{0,-\frac{\sqrt{2}}{m'\sqrt{\kappa+\kappa'}}\beta'}^\tr\in\D_\lft$
has norm
$\norm{w}_\lft = \frac{1}{\sqrt{m'\paren{\kappa+\kappa'}}}\beta'$,
we conclude that
$\norm{\Xi_{\TD}} \ge \norm{w}_\lft$
and hence
$\norm{\Xi_{\TD}} > 1$ for all $\beta'$ sufficiently large.
(Numerical experiments%
\footnote{We sampled 
$m'\in(0,10)$, 
$\kappa\in(0,1000)$, 
$\kappa'\in(0,1000)$, 
$\beta'\in(0,10)$
uniformly at random \numprint{100000} times.}
indicate that the induced norm is larger than unity for all $\beta' > 0$.)
We conclude that constraint activation is generally expansive.

\subsubsection{Elastic spring--damper with soft constraints}
\label{sec:appl:mech:soft}
The result in the previous section indicates that spring--damper networks subject to \emph{hard} unilateral constraints generally 
do not satisfy the discrete--time infinitesimal contractivity condition~\eqref{eq:thm:Xi} in Theorem~\ref{thm:main} 
when the system has more than a single degree--of--freedom
(the 1--DOF example from Sec.~\ref{sec:appl:mech:1dof} is contractive only because the constrained mode is zero--dimensional).
The reset, or restitution law, used to model impacts against hard constraints coarsely approximates the actual mechanics of the interaction between bodies, which consist of elastic and plastic deformation in the contact zone. An alternative approach is to explicitly model this deformation using additional forces.
In this section, we consider infinitesimal contractivity of the class of mechanical systems subject to \emph{soft} unilateral constraints studied in~\cite{Tornambe1999-dr}.
Specifically, rather than exactly enforcing unilateral constraints $a(q) \ge 0$ in~\eqref{eq:ddq}, we will \emph{penalize} constraint violation using a potential function that applies forces ``as though a linear elastic spring were located at the point of contact''~\cite[Sec~II-B]{Tornambe1999-dr} as illustrated in Fig.~\ref{fig:mech:soft}, 
yielding the modified potential energy
\eqnn{
\label{eq:soft:pe}
v_J(q) = \frac{1}{2} q^\tr K q + \frac{1}{2}\sum_{j\in J} k_j a_j(q)^2 = \frac{1}{2} q^\tr \paren{K + K_J} q 
}
where $K_J = Da_J^\tr \diag\set{k_j}_{j\in J} Da_J$ is a positive semidefinite stiffness matrix for the potential $v_J$ associated with the subset of constraints $J\subset\set{1,\dots,n}$ that are active, i.e. for which $a_J(q) \le 0$.
With this modification, the system's equations of motion become
\eqnn{
\label{eq:ddq:soft}
M\ddot{q} & = 
B u(t)
- (K + K_J) q
- C \dot{q},\ a_J(q) \le 0.
}
The dynamics in~\eqref{eq:ddq:soft} are classical in the sense that the right--hand side of the equation specifies a Lipschitz continuous and piecewise--differentiable vector field.
However, the natural distance metric determined by the sum of potential and kinetic energy,
\eqnn{
e_J(q,\dot{q}) = \frac{1}{2} q^\tr \paren{K + K_J} q + \frac{1}{2} \dot{q}^\tr M \dot{q},
}
now depends on the set of active constraints $J$,
whence the weighting matrix for the energy--induced $2$--norm also depends on the set of active constraints,
\eqnn{
E_J = D^2 e_J = \mat{cc}{K + K_J & 0 \\ 0 & M}.
}
Letting $x = (q,\dot{q})$ and $\dot{x} = \F_J(t,x)$ for $a_J(q) \le 0$, using the energy--induced $2$--norm yields a negative semidefinite matrix
\eqnn{
\frac{1}{2}\paren{D_x \F_J^\tr\cdot E_J + E_J \cdot D_x \F_J} & 
= \mat{cc}{0 & 0 \\ 0 & -B},
}
whereas using, e.g., the $2$--norm from the unconstrained mode yields
\eqnn{
\mat{cc}{q^\tr & \dot{q}^\tr} \paren{D_x \F_J^\tr\cdot E_\uncon + E_\uncon \cdot D_x \F_J} \mat{c}{q \\ \dot{q}} & 
= \mat{cc}{q^\tr & \dot{q}^\tr} \mat{cc}{0 & -\frac{1}{2} K_J  \\ -\frac{1}{2} K_J & -B} \mat{c}{q \\ \dot{q}} \\ 
& = -\dot{q}^\tr K_J q - 2\dot{q}^\tr B \dot{q},
}
which is not 
a negative semidefinite quadratic form in any constrained mode (i.e. when $J \neq \uncon$).
In the remainder of this section we will assess infinitesimal contractivity of~\eqref{eq:ddq:soft} using the energy--induced norms that depend on the set of active constraints. 

Infinitesimal contractivity for continuous and piecewise--differentiable vector fields where different norms (hence, matrix measures) are associated with each differentiable ``piece'' of the vector field have previously been considered in~\cite{Fiore:2016fj}, but only for \emph{switched} systems where the discrete transition between ``pieces'' is triggered by an exogenous input, i.e. does not depend on the continuous state. 
As the previous sections illustrate, 
the interaction between continuous--time~\eqref{eq:ddq:cont} and discrete--time~\eqref{eq:ddq:disc} dynamics can yield expansion even when the continuous and discrete components are individually non--expansive.
We will assess infinitesimal contractivity of this system 
by treating~\eqref{eq:ddq:soft} as a hybrid system and applying our results;
for ease of exposition, we will restrict our attention to the system from Section~\ref{sec:appl:mech:1dof} with $1$ DOF and $1$ constraint.

Noting that positions, velocities, and accelerations are continuous in~\eqref{eq:ddq:soft} when constraints (de)activate,
we conclude that both saltation matrices are the $2$--dimensional identity, $\Xi_{\uncon,\lft} = \Xi_{\lft,\uncon} = \I_2$.
Thus, the induced norms can be computed 
as
$\norm{\Xi_{\uncon,\lft}} = \sigma_{\max}\paren{S_\lft \Xi_{\uncon,\lft} S_\uncon^{-1}}$,
$\norm{\Xi_{\lft,\uncon}} = \sigma_{\max}\paren{S_\lft \Xi_{\uncon,\lft} S_\uncon^{-1}}$
with
\eqnn{
S_\uncon = \mat{cc}{\sqrt{\kappa} & 0 \\ 0 & \sqrt{m}},\ 
S_\lft = \mat{cc}{\sqrt{\kappa + \kappa''} & 0 \\ 0 & \sqrt{m}}
}
denoting the square roots of $E_\uncon$, $E_\lft$, respectively,
yielding
\eqnn{
\norm{\Xi_{\uncon,\lft}} &= \sigma_{\max}\paren{S_\lft S_\uncon^{-1}} = \sqrt{\frac{\kappa}{\kappa+\kappa''}} 
< 1,\\ 
\norm{\Xi_{\lft,\uncon}} &= \sigma_{\max}\paren{S_\uncon S_\lft^{-1}} 
= \norm{\Xi_{\uncon,\lft}}^{-1}
> 1.
}
We conclude constraint activation is contractive and constraint deactivation is expansive.

\subsubsection{Viscoelastic spring-dampers with $2$ DOF, $1$ constraint}
\label{sec:appl:mech:visco}
We now modify the example from Section~\ref{sec:appl:mech:2dof} by 
\rev{
adding a spring in series with the damper that connects the two masses, yielding a
\emph{visco}elastic spring-damper illustrated in Fig.~\ref{fig:mech:visco}; this choice is motivated by the following considerations. 
Coupling the two masses directly through a damper (i.e.\ without the series spring that appears in the viscoelastic model) causes discontinuous changes in force at the moment of impact.
Theoretically, although this discontinuity in force is unproblematic for contraction analysis of the impacting mass as seen in the preceding section, we saw in Section~\ref{sec:appl:mech:2dof} that the force discontinuity causes infinitesimal \emph{expansion} for the other mass.
Practically, the discontinuous change in force is physically implausible since the structures connecting bodies in real mechanical systems deform.
These theoretical and practical considerations justify our focus on viscoelastic connections between masses.
}%
This mechanical system subject to a unilateral constraint can be modeled in the hybrid systems framework from Section~\ref{sec:prelim}:
\begin{itemize}
\item[$\D$] $=  \D_{\uncon}\coprod \D_{\con}$ 
where 
$\D_{\uncon} = \set{(q,q',\ell,\dot{q},\dot{q}') \in\Rbb^5}$ 
is the set of continuous states wherein 
mass $m$ is unconstrained
and
$\D_{\con} = \set{(q',\ell,\dot{q}') \in \Rbb^3}$ 
is the set of continuous states wherein 
mass $m$ is constrained 
\rev{-- note the inclusion of $\ell$, the length state of the viscoelastic spring-damper};

\item[$\F$] $: [0,\infty)\times\D\into T\D$ 
defined by 
$\F_{\uncon}(t,q,q',\ell,\dot{q},\dot{q}') = (\dot{q},\dot{q}',\dot{\ell},\ddot{q},\ddot{q}')$
and
$\F_{\con}(t,q',\ell,\dot{q}') = (\dot{q}',\dot{\ell},\ddot{q}')$
where,
\rev{
with
\eqn{
v %
= 
\frac{1}{2} \left(
\kappa \left(u(t) - q\right)^{2} 
+ 
\kappa' \left(q' - q\right)^{2} 
+ 
\kappa'' \left(q' - q - \ell\right)^{2}
\right)
}
denoting the potential energy stored in springs,
\eqn{
\label{eq:ddq:visco}
\beta'\dot{\ell} = -D_\ell v,\ m\ddot{q} = -D_q v - \beta\dot{q},\ m'\ddot{q}' = -D_{q'} v;%
}%
}%
\item[$\G$] $= \G_{\TD}\coprod\G_{\LO}$ 
where 
\eqn{\G_{\TD} = \set{(t,q,q',\ell,\dot{q},\dot{q}')\in[0,\infty)\times\D_\uncon : q \le 0,\ \dot{q} < 0}}
is the set of 
states where mass $m$ impacts the stop 
and
\eqn{\G_{\LO} = \set{(t,q',\ell,\dot{q}')\in[0,\infty)\times\D : \ddot{q} \ge 0}}
is the set of 
states where mass $m$ accelerates off the stop;
 
\item[$\R$] $:\G\into\D$ defined by 
$\R|_{\G_{\TD}}(t,q,\dot{q},q',\dot{q}') = (q',\dot{q}')$ (i.e. the position and velocity of mass $m$ are ``forgotten'' since they are both equal to zero when the stop constraint activates) and
$\R|_{\G_{\LO}}(t,q',\dot{q}') = (0,q',0,\dot{q}')$ (i.e. the position and velocity of mass $m$ are reinitialized at zero when stop constraint deactivates).

\end{itemize}

As with the previous example, we let the sum of potential and kinetic energy determine the weighting matrices for $2$-norms that will be used to assess infinitesimal contractivity of continuous-time flow and discrete-time reset.
In the unconstrained mode, the energy is
\eqnn{
\nonumber e_{\uncon} &=
\rev{v}
+
\frac{\dot{q}^{2} m}{2} 
+
\frac{\dot{q}'^{2} m'}{2} 
}
yielding metric
\eqnn{
E_{\uncon} = D^2 e_{\uncon}
= 
\left[\begin{matrix}\kappa + \kappa' + \kappa'' & - \kappa' - \kappa'' & \kappa'' & 0 & 0\\- \kappa' - \kappa'' & \kappa' + \kappa'' & - \kappa'' & 0 & 0\\\kappa'' & - \kappa'' & \kappa'' & 0 & 0\\0 & 0 & 0 & m & 0\\0 & 0 & 0 & 0 & m'\end{matrix}\right].
}
\rev{
The energy $e_{\con}$ in the constrained mode does not vary with $q$ or $\dot{q}$, 
\eqnn{
e_{\con} &=
\frac{\kappa u(t)^{2}}{2} 
+ 
\frac{\kappa' q'^{2}}{2} 
+ 
\frac{\kappa''}{2} \left(q' - \ell\right)^{2}
+
\frac{\dot{q}'^{2} m'}{2} 
}
yielding metric
\eqnn{
E_{\con} =  D^2 e_{\con}
&=
\left[\begin{matrix}\kappa' + \kappa'' & - \kappa'' & 0\\- \kappa'' & \kappa'' & 0\\0 & 0 & m'\end{matrix}\right].
}
}

We first consider infinitesimal contractivity of continuous-time flow.
Letting $x = (q,q',\ell,\dot{q},\dot{q}')$ denote the continuous state vector in the unconstrained mode and $\dot{x} = \F_{\uncon}(t,x)$ denote its time derivative yields
\eqnn{
\nonumber &\frac{1}{2}\paren{ D_x \F_{\uncon}^\tr \cdot E_{\uncon} + E_{\uncon} \cdot D_x \F_{\uncon} } =\\
&\diag\paren{
\left[\begin{matrix}- \frac{\kappa''^{2}}{\beta'} & \frac{\kappa''^{2}}{\beta'} & - \frac{\kappa''^{2}}{\beta'}\\\frac{\kappa''^{2}}{\beta'} & - \frac{\kappa''^{2}}{\beta'} & \frac{\kappa''^{2}}{\beta'}\\- \frac{\kappa''^{2}}{\beta'} & \frac{\kappa''^{2}}{\beta'} & - \frac{\kappa''^{2}}{\beta'}\end{matrix}\right]
,
\left[\begin{matrix}- \beta & 0\\0 & 0\end{matrix}\right]
}.
}
Similarly, letting $x = (q',\ell,\dot{q}')$ denote the continuous state vector in the constrained mode and $\dot{x} = \F_{\con}(t,x)$ denote its time derivative yields
\eqnn{
\nonumber &\frac{1}{2}\paren{D_x \F_{\con}^\tr \cdot E_{\con} + E_{\con} \cdot D_x \F_{\con}} \\
&=
\diag\paren{
\left[\begin{matrix}- \frac{\kappa''^{2}}{\beta'} & \frac{\kappa''^{2}}{\beta'}\\\frac{\kappa''^{2}}{\beta'} & - \frac{\kappa''^{2}}{\beta'}\end{matrix}\right]
,
0
}.
}
The spectrum of 
$\frac{1}{2}\paren{D_x \F_{\uncon}^\tr \cdot E_{\uncon} + E_{\uncon} \cdot D_x \F_{\uncon}}$
is 
$\set{0,-\beta,-\frac{3\kappa''^2}{\beta'}}$
and
that of
$\frac{1}{2}\paren{D_x \F_{\con}^\tr \cdot E_{\con} + E_{\con} \cdot D_x \F_{\con}}$
is 
$\set{0,-\frac{2\kappa''^2}{\beta'}}$, 
so the matrix measures of $D_x\F_\uncon$ and $D_x\F_\con$ are both equal to 0 (zero),
whence continuous-time flow is non-expansive%
\footnote{Since some eigenvalues are negative, 
the flow is 
\emph{semi-}~\cite[Sec.~2]{Lohmiller1998-xj}
or 
\emph{horizontally}~\cite[Sec.~VII]{Forni2014-pq}
contractive.}
in both constrained and unconstrained modes.

We now consider infinitesimal contractivity of discrete-time reset, i.e. we evaluate induced norms of saltation matrices.
The \emph{liftoff} (\LO) saltation matrix is
\eqnn{
\Xi_{\LO} &=
D_x \R|_{\G_{\LO}} =
\left[\begin{matrix}0 & 0 & 0\\1 & 0 & 0\\0 & 1 & 0\\0 & 0 & 0\\0 & 0 & 1\end{matrix}\right];
}
since this matrix is an energy-preserving embedding, it is an isometry with respect to the 2-norms determined by the energy metric, hence its induced norm is equal to unity.
Similarly, the \emph{touchdown} (\TD) saltation matrix,
\eqnn{
\Xi_{\TD} &=
D_x \R|_{\G_{\TD}} =
\left[\begin{matrix}0 & 1 & 0 & 0 & 0\\0 & 0 & 1 & 0 & 0\\0 & 0 & 0 & 0 & 1\end{matrix}\right],
}
is an orthogonal projection with respect to the 2-norms determined by the energy metric, hence its induced norm is equal to unity.

Combining our analyses of infinitesimal contractivity of continuous-time flow and discrete-time reset from the two preceding paragraphs, we conclude that Theorem~\ref{thm:main} applies to this system with 
$c = 0$ in~\eqref{eq:thm:mu} 
and 
$K = 1$ in~\eqref{eq:thm:Xi},
i.e. the system's dynamics are non-expansive.

\subsubsection{\mbox{Viscoelastic spring-dampers with $3$ DOF, $2$ constraints}}
\label{sec:appl:mech:visco32}

\rev{
We now extend the example from the preceding section by adding a second identical viscoelastic ``limb'' attached to the same ``body'' mass as in Fig.~\ref{fig:mech:visco32}.
Since the two limbs can impact at arbitrary times with respect to one another, this model restricts neither dwell time nor mode sequence.
Nevertheless, the following analysis will show that the hypotheses of our Theorem~\ref{thm:main} are satisfied,
demonstrating the key contributions of our Theorem~\ref{thm:main} with respect to prior literature on contraction analysis for hybrid systems.
This mechanical system subject to unilateral constraints can be modeled in the hybrid systems framework from Section~\ref{sec:prelim} as a straightforward extension of the viscoelastic model in the preceding section.
Since the system has 8 states, 4 modes, and 12 guard/reset pairs, we will not tax the reader's patience by explicitly specifying each component of the hybrid system.
Instead, in the interest of readability, we will specify the essential features of the model and summarize the results of applying our infinitesimal contraction analysis.

This system has four modes, 
$\J = \set{\nth,\lft,\rght,\bth}$; 
with 
$q\in\Rbb^3$ 
and
$\ell\in\Rbb^2$ 
denoting configurations of 3 masses and 2 viscoelastic dampers,
$\D_{\nth} = \set{(q,\ell,\dot{q}) \in\Rbb^8}$ 
is the set of continuous states wherein 
both limb masses are unconstrained
and 
$\D_{\lft}$,
$\D_{\rght}$,
$\D_{\bth}$
are the sets of continuous states wherein the the left, right, or both limb masses are constrained (i.e. zero position and velocity), respectively.

The vector field is defined by 
$(\dot{q},\dot{\ell},\ddot{q})$
where,
with $v$ denoting potential spring energy 
and
$M = \diag\paren{m',m,m}$,
$C = \diag\paren{0,\beta,\beta}$
denoting mass and damping matrices,
\eqnn{
\label{eq:ddq:visco32}
\beta'\dot{\ell} = -D_\ell v,\ M\ddot{q} = B u(t) -D_q v - C\dot{q}.
}

Since the limb masses are decoupled through the body mass, each of the 4 modes can transition to each of the 3 other modes, so there are $4\times 3 = 12$ distinct pieces of the guard set and reset function corresponding to every pair $j,j'\in\J$ with $j\ne j'$.
However, we note that the inequalities and functions that define constraint activation and deactivation for one limb are unaffected by the state of the other limb:
each limb mass transitions from unconstrained to constrained when it impacts the stop as in~\eqref{eq:TD} and transitions from constrained to unconstrained when it accelerates off the stop as in~\eqref{eq:LO}.
This inertial decoupling is necessary to ensure the system's flow is continuous, i.e.\ to satisfy Assumption~\ref{assum:flow}.2; we refer the interested reader to~\cite[Thm.~20]{Ballard2000-ui} for more details.

Using the total energy $e_j = v_j + \frac{1}{2} \dot{q_j}^\top M \dot{q_j}$ to determine the weighting matrix $E_j = D^2 e_j$ for a 2-norm in each mode $j\in\J$, 
we find%
\eqnn{%
\begin{tabular}{c|c}
$j$ & 
$\spec\frac{1}{2}\paren{D_x \F_{j}^\tr \cdot E_{j} + E_{j} \cdot D_x \F_{j}}$ 
\\
\hline
$\nth$ & 
$\set{0, -\beta, -\frac{4 \kappa''^2}{\beta'}, -\frac{2 \kappa''^2}{\beta'}}$ 
\\
$\lft,\rght$ & 
$\set{0, -\beta, -\frac{(5 \pm \sqrt{5})\kappa''^2}{2\beta'}}$ 
\\
$\bth$ & 
$\set{0, -\frac{3 \kappa''^2}{\beta'}, -\frac{\kappa''^2}{\beta'}}$ 
\end{tabular}
}%
so the matrix measure of each vector field derivative is equal to 0 (zero), whence continuous-time flow is non-expansive.%
\footnote{\rev{As with the example in Section~\ref{sec:appl:mech:visco}, 
the flow is 
\emph{semi-}~\cite[Sec.~2]{Lohmiller1998-xj}
or 
\emph{horizontally}~\cite[Sec.~VII]{Forni2014-pq}
contractive.}}
Evaluating the saltation matrices for every pair $j,j'\in\J$ with $j\neq j'$, we find that these matrices are energy-preserving embeddings or projections (depending on whether $j$ has more or fewer active constraints than $j'$), hence their induced norms are equal to unity, whence discrete-time reset is non-expansive.

Combining our analyses of infinitesimal contractivity of continuous-time flow and discrete-time reset from the preceding paragraph, we conclude that Theorem~\ref{thm:main} applies to this system with $c = 0$ in~\eqref{eq:thm:mu} and $K = 1$ in~\eqref{eq:thm:Xi}, i.e.\ the system's dynamics are non-expansive.
Since the two limbs can impact at arbitrary times with respect to one another, this model restricts neither dwell time nor mode sequence,
demonstrating the key contributions of our Theorem~\ref{thm:main} with respect to prior literature on contraction analysis for hybrid systems.
}

\subsubsection{Viscoelastic spring-dampers with $d$ DOF, $n$ constraints}
\label{sec:appl:mech:viscodn}

We now generalize the example from the preceding section 
to the class of linear ``MCK'' systems subject to linear unilateral constraints specified as above at the outset of Sec.~\ref{sec:appl:mech}.
Specifically, we begin with dynamics
\eqnn{
M \ddot{q} + C \dot{q} + K q = B u(t)\ \text{s.t.}\  
a(q) = A q \ge 0
}
and assume, as above, that
$M = M^\top > 0$,
$C = C^\top \geq 0$,
$K = K^\top \geq 0$,
and
$A\ M^{-1} A^\top = I$.
To generalize the construction from Sec.~\ref{sec:appl:mech:visco}, where elastic spring-dampers were replaced with viscoelastic elements, we introduce new 
``damper length'' states $l\in\mathbb{R}^{d\cdot(d+1)/2}$ -- one for each damper.
The new springs we introduce (in series with the dampers) yield a modified potential energy
\eqnn{\label{sec:appl:mech:viscodn:v}
\md{v} = \frac{1}{2} \mat{cc}{\ell^\top & q^\top} \mat{cc}{\md{K}_L & \md{K}_{LQ} \\ \md{K}_{QL} & \md{K}_Q} \mat{c}{\ell \\ q} = \frac{1}{2} \md{q}^\top \md{K} \md{q} 
}
where $\md{K} = \md{K}^\top > 0$ so that the modified total energy $\md{e} = \frac{1}{2} q^\top M q + \md{v}$ determines a positive-definite weighting matrix for a 2-norm,
\eqnn{\label{eq:appl:mech:viscodn:E}
\md{E} = D^2 \md{e} = \mat{cc}{\md{K} & 0 \\ 0 & M} = \mat{ccc}{\md{K}_L & \md{K}_{LQ} & 0 \\ \md{K}_{QL} & \md{K}_Q & 0 \\ 0 & 0 & M}.
}
Replacing spring-dampers with viscoelastic elements yields the modified dynamics
\eqnn{\label{eq:appl:mech:viscodn:ddq}
\md{C} \dot{\ell} &= - D_\ell \md{v}^\top = - \md{K}_L \ell - \md{K}_{LQ} q, \\
M \ddot{q} &= B u(t) - D_q \md{v}^\top = B u(t) - \md{K}_{QL} - \md{K}_Q q \ell,
}
where $\md{C} = \diag\set{\beta_{ij} : 1 \leq i \leq j \leq d}$ is a diagonal matrix constructed from the (positive) viscosity parameters $\beta_{ij} > 0$; note, in particular, that $\md{C} = \md{C}^\top > 0$.

We first consider infinitesimal contractivity of continuous-time flow.
Letting $x = (\ell,q,\dot{q})$ denote the continuous state vector in the unconstrained mode and $\dot{x} = \F_\uncon(t,x)$ denote its time derivative yields
\eqnn{
\frac{1}{2}\paren{ D_x \F_{\uncon}^\tr \cdot \md{E} + \md{E} \cdot D_x \F_{\uncon} } 
&= \mat{ccc}{ -\md{K}_L \md{C}^{-1} \md{K}_L & -\md{K}_L \md{C}^{-1} \md{K}_{LQ} & 0 \\ - \md{K}_{QL} \md{C}^{-1} \md{K}_L & -\md{K}_{QL} \md{C}^{-1} \md{K}_{LQ} & 0 \\ 0 & 0 & 0} \\
& = -\mat{c}{\md{K}_L \\ \md{K}_{QL} \\ 0} \md{C}^{-1} \mat{ccc}{\md{K}_L & \md{K}_{LQ} & 0} \leq 0.
}
Since this matrix is negative semidefinite, the matrix measure of $D_x\F_\uncon$ is equal to $0$ (zero),
whence continuous-time flow is non-expansive%
\footnote{Since some eigenvalues are negative, 
the flow is 
\emph{semi-}~\cite[Sec.~2]{Lohmiller1998-xj}
or 
\emph{horizontally}~\cite[Sec.~VII]{Forni2014-pq}
contractive.}.

We now consider infinitesimal contractivity of discrete-time reset, i.e.\ we evaluate induced norms of saltation matrices.
Evaluating the saltation matrices for every pair $j,j'\in\J$ with $j\neq j'$, we find that these matrices are energy-preserving embeddings or projections (depending on whether $j$ has more or fewer active constraints than $j'$), hence their induced norms are equal to unity, whence discrete-time reset is non-expansive.

Combining our analyses of infinitesimal contractivity of continuous-time flow and discrete-time reset from the two preceding paragraphs, we conclude that Theorem~\ref{thm:main} applies to this system with 
$c = 0$ in~\eqref{eq:thm:mu} 
and 
$K = 1$ in~\eqref{eq:thm:Xi},
i.e. the system's dynamics are non-expansive.

\section{\rev{Necessity}}
\label{sec:necessity}
We now present a %
result %
indicating that the infinitesimal continuous-time~\eqref{eq:thm:mu} and discrete-time~\eqref{eq:thm:Xi} contraction conditions of Theorem \ref{thm:main} are 
\rev{necessary for strict contraction with respect to the intrinsic distance function~\eqref{eq:thm:d}.}

\begin{thm}
\label{thm:necessity}
Consider a hybrid system satisfying Assumptions~\ref{assum:system} and~\ref{assum:flow} and further assume 
(a)
no guards overlap so that 
$\ol{\G}_{j,j'} \cap \ol{\G}_{j,j''} = \emptyset$ for all $j'\ne j''$
and 
(b) no mode is ever entirely guard so that $\D_j\backslash \G_j(t)\neq \emptyset$ for all $j$ and all $t\geq 0$.
If the hybrid system is contractive, 
i.e., 
if there exists a 
constant $c\in\mathbb{R}$ such that%
\footnote{$d_t:\mcMt\times\mcMt\into[0,\infty)$ 
is the specific (time-varying) intrinsic distance function defined in Sec.~\ref{sec:prelim:distance} on the hybrid system's quotient space $\mcMt$.}
\begin{align}
  \label{eq:2}
d_t(\phi(t,s,\xi),\phi(t,s,\zeta))\leq e^{c(t-s)}d_s(\xi,\zeta) \end{align} 
for all $\xi,\zeta\in \D$, $t \ge s\ge 0$, 
then the 
continuous-time~\eqref{eq:thm:mu} 
and 
discrete-time~\eqref{eq:thm:Xi} 
contraction conditions in Theorem~\ref{thm:main} are satisfied:
\begin{align}
  \label{eq:22}
  \mu_j\left(D_x \F_j(t,x)\right)\leq c
\end{align}
for all $t\geq 0$ and all $x\in\D\backslash \G(t)$, and 
\begin{align}
  \label{eq:15}
\|\Xi(t,x)\|_{j,j'}\leq 1  
\end{align}
for all $j,j'\in\J$, all $t> 0$, and all $x\in\G_{j,j'}(t)$.%

\end{thm}
\begin{proof}
Fix a time $\sigma\geq 0$ and consider $x\in\D_j\backslash \G(\sigma)$ for some $j\in \J$, and recall that $\D_j\backslash\G(\sigma)$ is open. Let $\delta>0$ be such that $\mathcal{B}_\delta=\{\zeta:|\zeta-x|_j<\delta\}\subset \D_j\backslash \G(\sigma)$, and notice that $d_\sigma(x,\zeta)=|x-\zeta|_j$ for all $\zeta\in\mathcal{B}_\delta$. 
By standard converse results for continuously differentiable contracting systems, \emph{e.g.}, \cite[Proposition 3]{Aminzare:2014rm}, $\mu_j\left(D_x \F_j(\sigma,x)\right)\leq c$. Since $\sigma$, $x$, and $j$ where arbitrary, \eqref{eq:22} holds.

Now fix a time $\sigma >0$ and consider $x^*\in \G_{j,j'}(\sigma)$ for some $j,j'\in \D$. By Assumption \ref{assum:system}.4, $D_tg_{j,j'}(\sigma,x^*)+D_xg_{j,j'}(\sigma,x^*)\cdot \F_j(\sigma,x^*)<0$. Since also $\D_j\backslash \G_{j}(\sigma)\neq \emptyset$, there exists $x_0\in \D_j$ and time $\tau<\sigma$ satisfying $\phi_j(\sigma^-,\tau,x_0)=x^*$. Further, because guards do not overlap, there exists an open neighborhood $\mathcal{O}\ni x_0$ such that for all $\xi\in\mathcal{O}$, there exists $\theta(\xi)>\tau$ such that $\phi(\theta(\xi)^-,\tau,\xi)\in \G_{j,j'}(\theta(\xi))$ and $\phi(t,\tau,\xi)\in D_j\backslash \G(t)$ for all $t\in[\tau,\theta(\xi))$.  We write $x(t)=\phi(t,\tau,x_0)$ so that, in particular, $x(\tau)=x_0$ and $x(\sigma^-)=x^*$. 

Assume the neighborhood $\mathcal{O}$ is chosen small enough so that there exists $\bar{\tau}>0$ such that $\phi(t,\tau,\xi)\in \D_{j'}$ for all $t\in[\theta(\xi),\sigma+\bar{\tau})$ for all $\xi\in \mathcal{O}$; %
existence of such a $\bar{\tau}$ for small enough $\mathcal{O}$ is guaranteed by Assumption \ref{assum:system}.2. For each $\epsilon\in(0,\min\{\sigma-\tau, \bar{\tau}-\sigma\})$, %
 let $\delta(\epsilon)>0$ be small enough so that, for all $\xi\in \mathcal{B}_{\delta(\epsilon)}(\sigma-\epsilon, x(\sigma-\epsilon))=\{\xi\in \D: d_{\sigma-\epsilon}(\xi, x(\sigma-\epsilon))<\delta(\epsilon)\}$, it holds that %
\begin{align}
  \label{eq:11}
&  d_{\sigma-\epsilon}(\xi, x(\sigma-\epsilon))=|\xi-z(\sigma-\epsilon)|_j,\\
\nonumber &d_{\sigma+\epsilon}(\phi(\sigma+\epsilon,\sigma-\epsilon,\xi),\phi(\sigma+\epsilon,\sigma-\epsilon, x(\sigma-\epsilon)))\\
\label{eq:11-2}&=|\phi(\sigma+\epsilon,\sigma-\epsilon,\xi)-\phi(\sigma+\epsilon,\sigma-\epsilon, x(\sigma-\epsilon))|_{j'}.
\end{align}
Then the proof is complete by observing that
\begin{align}
\nonumber   1&\geq \limsup_{\epsilon\to 0^+}\sup_{\xi\in\mathcal{B}_{\delta(\epsilon)}(x(\sigma-\epsilon))}\\
  \label{eq:27}&\qquad \frac{|\phi(\sigma+\epsilon,\sigma-\epsilon,\xi)-\phi(\sigma+\epsilon,\sigma-\epsilon, x(\sigma-\epsilon))|_{j'}}{|\xi-x(\sigma-\epsilon)|_j}\\
\nonumber &\geq \limsup_{\epsilon\to0^+}\sup_{\xi\in\mathbb{R}^{n_j}}\\
&\qquad \frac{|D_3\phi(\sigma+\epsilon,\sigma-\epsilon,x(\sigma-\epsilon))\cdot \xi|_{j'}}{|\xi|_j}\\
&=|\Xi(\sigma,x^*)|_{j,j'},
\end{align}
where the second inequality follows by the definition of directional derivative and the first inequality follows from the following observation combined with \eqref{eq:11} and \eqref{eq:11-2}:
from \eqref{eq:2}, for any $\xi,\zeta \in \D$ and all $\epsilon\in(0,\sigma]$,
\begin{align}
  \label{eq:5}
\frac{ d_{\sigma+\epsilon}(\phi(\sigma+\epsilon,\sigma-\epsilon,\xi),\phi(\sigma+\epsilon,\sigma-\epsilon,\zeta))}{d_{\sigma-\epsilon}(\xi,\zeta)}\leq e^{2c\epsilon}
\end{align}
whenever $d_{\sigma-\epsilon}(\xi,\zeta)\neq 0$, so that, in particular,
\begin{align}
\nonumber &\limsup_{\epsilon\to 0^+}\sup_{\xi,\zeta\in\D}\frac{ d_{\sigma+\epsilon}(\phi(\sigma+\epsilon,\sigma-\epsilon,\xi),\phi(\sigma+\epsilon,\sigma-\epsilon,\zeta))}{d_{\sigma-\epsilon}(\xi,\zeta)}\\
  \label{eq:9}&\leq \lim_{\epsilon\to 0^+}e^{2c\epsilon}=  1.
\end{align}

\end{proof}

\begin{rem}[summary of necessity result]
By restricting our analysis to  hybrid systems whose guards are codimension-1 submanifolds in the above Theorem, we find that contraction with respect to the intrinsic distance function defined in Section~\ref{sec:prelim:distance} implies infinitesimal contraction in continuous~\eqref{eq:thm:mu} and discrete~\eqref{eq:thm:Xi} time.
Extensions of this result to more general guard structures are discussed in Section~\ref{sec:disc:necessity}.
\end{rem}

\section{Discussion}
\label{sec:disc}
Before concluding, we discuss additional connections and possible extensions or applications of the preceding results.

\subsection{Assumptions~\ref{assum:system} and~\ref{assum:flow}}
\label{sec:disc:assum}
We acknowledge that the Assumptions imposed in Section~\ref{sec:prelim} are %
not obviously satisfied in all applications of interest.
However, as discussed in the sequence of Remarks that followed the Assumptions, each condition plays a crucial role in the Proof of Theorem~\ref{thm:main}; if any one condition is violated, then there exist systems that cannot be shown contractive using our approach.
Furthermore, recent results provide a broad class of systems that satisfy these Assumptions.
For instance,~\cite{Burden2016-bb} considers a class of 
discontinuous vector fields that yield continuous and piecewise-differentiable flows.
Hybrid systems obtained from such vector fields satisfy 
Assumption~\ref{assum:system} by construction~\cite[Def.~2]{Burden2016-bb} 
and
Assumption~\ref{assum:flow} by piecewise-differentiability of the flow~\cite[Thm.~5]{Burden2016-bb};
these observations justify application of our results in 
Section~\ref{sec:appl:traffic}.
Related work~\cite{Pace2017-tt} established conditions under which the hybrid system model of a mechanical system subject to unilateral constraints has a continuous and piecewise-differentiable flow.
The conditions on the hybrid system~\cite[Assumps.~1-4]{Pace2017-tt} and properties of its flow~\cite[Thm.~1]{Pace2017-tt} ensure that Assumptions~\ref{assum:system} and~\ref{assum:flow} are satisfied for this class of systems,
justifying application of our results to the examples in Section~\ref{sec:appl:mech}.
Broadening the class of systems that are known to satisfy these Assumptions is the subject of ongoing work.

\rev{%
\subsection{Sliding modes}
\label{sec:disc:sliding}
A piecewise-smooth vector field exhibits \emph{sliding} when it is discontinuous along a submanifold and points toward the submanifold from all sides~\cite{Jeffrey2014-nt}; in this case, trajectories can be defined within this \emph{sliding mode} using Carath{\'{e}}odory solutions~\cite{Filippov1988-nh}.
Assumption~{\ref{assum:system}.1} precludes sliding in the vector field component of the hybrid systems studied here.
However, if the sliding mode is regarded as a distinct mode in a hybrid system that satisfies Assumptions~{\ref{assum:system}} and~{\ref{assum:flow}}, then our techniques could be applied.
We conclude by noting that~\cite{Fiore:2016fj} provides an alternative approach to assessing infinitesimal contractivity of sliding vector fields.
}

\rev{
\subsection{Periodic orbits}
\label{sec:disc:periodic}
Establishing the existence or stability of periodic orbits is a common application of contraction analysis.
When a system is contractive as in~\eqref{eq:thm:d} on a complete forward-invariant set for a given periodic input, invoking the Banach contraction mapping principle establishes the existence of a stable periodic orbit~\cite[Thm.~2]{Sontag2010-qg}.
When contraction analysis is applied in this way, the existence of the periodic orbit and its stability are both consequences of contractivity.
For instance, numerical simulations indicate that periodic forcing yields stable periodic orbits in the mechanical systems from Section~\ref{sec:appl:mech}.

When the existence of a periodic orbit is established or assumed \emph{a priori} in a time-invariant system, contraction analysis can be applied to assess stability of the orbit.
However, in this case, it is important to note that perturbations in \emph{phase} persist, so the system can only be contractive in directions \emph{transverse} to the periodic orbit.
Transverse contractivity criteria for periodic orbits in hybrid systems was developed in~\cite{Tang2014-rd} with restrictions on mode sequence and dwell time; these restrictions arise implicitly from the assumption that the periodic orbit passes transversally through a sequence of codimension-1 guards~\cite[Sec.~II and Fig.~1]{Tang2014-rd}.
Generalizing our results to accommodate
\emph{semi-}~\cite[Sec.~2]{Lohmiller1998-xj}
or 
\emph{horizontal}~\cite[Sec.~VII]{Forni2014-pq}
contraction as discussed in the following section provides one possible approach to derive transverse contractivity criteria for periodic orbits in hybrid systems without restricting dwell time or mode sequence.
}

\subsection{Generalizations of infinitesimal contraction}
\label{sec:disc:generalizations}

This paper focused on the strongest possible notion of contraction, requiring that the distance between trajectories decrease in all directions at every instant.
Recent work has considered relaxations of these strict requirements, 
indicating possible routes to extend our approach.
Importantly, although we had to introduce new and non-trivial analysis techniques to generalize infinitesimal contraction to hybrid systems (specifically, the intrinsic distance function and quotient space from Section~\ref{sec:prelim:distance}), 
the core of our approach leverages the same intuition 
utilized in classical systems (i.e. differential or difference equations, exclusive ``or''):
a system is contractive if path lengths decrease in time\rev{~\cite{Lohmiller1998-xj,Sontag2010-qg,Forni2014-pq}}.
This close parallel may facilitate generalization of a variety of classical techniques to the hybrid setting
\rev{including state-varying distance metrics defined by Riemannian~\cite{Lohmiller1998-xj} or Finsler~\cite{Forni2014-pq} structures}.

\rev{
\subsection{Incremental stability}
\label{sec:disc:inc}

In systems governed by a single differential or difference equation, the concept of \emph{contraction} is intimately related to the concept of \emph{incremental stability} --
both concepts formalize how discrepancies in initial conditions propagate through time.
In hybrid systems, where both state and time generally have continuous and discrete elements, the technical definition of these concepts varies depending on how one chooses to quantify the propagation of state discrepancies through time.
For instance, approaches based on graphical analysis of trajectory closeness implicitly disallow discrepancies in 
discrete elements of state~\cite{Li2016-zr}%
\footnote{\rev{This issue is discussed in more detail elsewhere~\cite[Sec.~V-A]{Burden2015-ip}.}},
or focus on propagation through discrete \emph{or} continuous time due to the restrictiveness of considering sumultanous propagation in discrete \emph{and} continuous time~\cite[Prop.~1]{Biemond2018-dx}.
Our use of an intrinsic distance function on the hybrid system's state space yields a simple connection between the concepts: 
Theorem~\ref{thm:main} establishes that
\emph{infinitesimal contractivity} (\eqref{eq:thm:mu} and~\eqref{eq:thm:Xi})
implies \emph{contractivity} \eqref{eq:thm:d},
which is simply a strong (exponential) form of \emph{incremental global asymptotic stability}~\cite[Def.~2.1]{Angeli2002-zx} stated in terms of the intrinsic distance function rather than the distance metric induced by a single vector space norm.
}

\subsection{Extending \rev{necessity to higher codimension}}
\label{sec:disc:necessity}
In Theorem~\ref{thm:necessity}, by restricting our attention to infinitesimal contraction through a codimension-1 guard, we found that the \rev{infinitesimal} continuous-time~\eqref{eq:thm:mu} and discrete-time~\eqref{eq:thm:Xi} contraction conditions of Theorem~\ref{thm:main} are \rev{necessary for strict contraction with respect to the intrinsic distance function~\eqref{eq:thm:d}}.
In systems with overlapping guards so that $\ol{\G}_{j,j'}\cap\ol{\G}_{j,j''} \ne \emptyset$ for some $j'\ne j''$, 
this necessity result applies on any codimension-1 portions of $\G_{j,j'}$ and $\G_{j,j''}$; continuous differentiability of the vector field, guard, and reset ensures the conclusions of the necessity result extend to the closure of these codimension-1 sets.
If a guard $\G_{j,j'''}$ is contained entirely within a codimension-$k$ set where $k > 1$, then the proof of Theorem~\ref{thm:necessity} can be adapted to establish contraction conditions on the vector field derivative and saltation matrix operator in directions tangent to $\G_{j,j'''}$ \rev{as in~\cite{Tang2014-rd}}.

\section{Conclusion}
\label{sec:conclusion}
We generalized infinitesimal contraction analysis to hybrid systems by leveraging local properties of continuous-time flow and discrete-time reset to bound the time rate of change of the intrinsic distance between trajectories.
\rev{Furthermore}, we showed that contraction with respect to this intrinsic distance function implies infinitesimal contraction in continuous and discrete time.
In addition to expanding the toolkit for stability analysis in hybrid systems, we provide novel bounds for the intrinsic distance function even when the system is not contractive.

\rev{
\section*{Acknowledgments}
We thank Matthew D.\ Kvalheim for correcting an error in an earlier version of this manuscript and four anonymous reviewers for providing invaluable feedback.
}

\bibliographystyle{ieeetr}
\bibliography{refs}

\end{document}